\documentclass[a4paper,12pt]{amsart}
\usepackage{microtype}

\usepackage[textwidth=16cm,textheight=23cm]{geometry}

\usepackage{graphicx}
\usepackage{amsmath,amssymb,amsthm}
\usepackage{bm}

\usepackage{longtable}
\usepackage{pdflscape} 

\usepackage[format=plain,justification=raggedright,singlelinecheck=false,font=small,labelfont=bf,labelsep=space]{caption}

\usepackage[caption=false,position=top]{subfig}
\usepackage{etoolbox}

\usepackage{url} 
\usepackage{thmtools,thm-restate}
\usepackage{array}
\usepackage{bbold}
\usepackage[dvipsnames,table]{xcolor}
\usepackage{booktabs}

\usepackage{verbatim}
\usepackage{rotating}

\usepackage{mathtools}
\DeclarePairedDelimiterXPP\seq[2]{}{\big(}{\big)}{_{#2}}{#1}

\newtheorem{prop}{Proposition}

\newtheorem{defn}{Definition}
\newtheorem{remarks}{Remarks}
\def\scal#1#2{\langle #1\mid #2 \rangle}

\def\scal#1#2{\langle #1\bv#2 \rangle}

\def\bv{\mid}

\makeatletter
\def\bstr{b}
\def\bfstr{bf}
\def\cstr{c}
\def\fstr{f}
\def\strLst{A,B,C,D,d,E,F,G,H,I,J,K,L,M,N,O,P,Q,R,S,T,U,V,W,X,Y,Z}

\newcommand{\MkB}[1]{\expandafter\def\csname\bstr#1\endcsname{\mathbb{#1}}}
  \@for\i:=\strLst\do{%
    \expandafter\MkB \i     }
    
\newcommand{\MkBF}[1]{\expandafter\def\csname\bfstr#1\endcsname{\mathbf{#1}}}
  \@for\i:=\strLst\do{%
    \expandafter\MkBF \i     }

\newcommand{\MkCal}[1]{\expandafter\def\csname\cstr#1\endcsname{\mathcal{#1}}}
  \@for\i:=\strLst\do{%
    \expandafter\MkCal \i     }
    
\newcommand{\MkFrak}[1]{\expandafter\def\csname\fstr#1\endcsname{\mathfrak{#1}}}
  \@for\i:=\strLst\do{%
    \expandafter\MkFrak \i     }
    
\makeatother

\newcommand{\Mod}[1]{\ \mathrm{mod}\ #1}

\newmuskip\pFqmuskip

\newcommand*\pFq[6][8]{%
  \begingroup 
  \pFqmuskip=#1mu\relax
  \mathchardef\normalcomma=\mathcode`,
  \mathcode`\,=\string"8000
  \begingroup\lccode`\~=`\,
  \lowercase{\endgroup\let~}\pFqcomma
  {}_{#2}F_{#3}{\left[\genfrac..{0pt}{}{#4}{#5};#6\right]}%
  \endgroup
}
\newcommand{\pFqcomma}{{\normalcomma}\mskip\pFqmuskip}

\makeatletter

\newcommand*{\@rowstyle}{}

\newcommand*{\rowstyle}[1]{
  \gdef\@rowstyle{#1}%
  \@rowstyle\ignorespaces%
}

\newcolumntype{=}{
  >{\gdef\@rowstyle{}}%
}

\newcolumntype{+}{
  >{\@rowstyle}%
}

\makeatother

\usepackage{listings}

\lstdefinelanguage{Maple}%
{morekeywords={and,assuming,break,by,catch,description,do,done,%
elif,else,end,error,export,fi,finally,for,from,global,if,%
implies,in,intersect,local,minus,mod,module,next,not,od,%
option,options,or,proc,quit,read,return,save,stop,subset,then,%
to,try,union,use,uses,while,xor},%
sensitive=true,%
morecomment=[l]\#,%
morestring=[b]",%
morestring=[d]"%
}[keywords,comments,strings]%

\definecolor{Frame}{RGB}{55,155,198}

\newcolumntype{L}{>{$}l<{$}} 

\newcolumntype{R}{>{$}r<{$}} 

\newcolumntype{C}{>{$}c<{$}} 

\usepackage{hyperref}

\begin{document}

\title[Lacunary generating functions of Hermite polynomials]{Explicit formulae for all higher order\\ exponential lacunary generating functions\\ of Hermite polynomials}

\author{N.~Behr}
\address{Universit\'{e} Paris-Diderot (Paris 07), IRIF, F-75205 Paris Cedex 13, France}
\email{nicolas.behr@irif.fr}
\thanks{Corresponding author email address: \texttt{nicolas.behr@irif.fr}. The work of NB was supported by a \emph{Marie Sk\l{}odowska-Curie actions Individual Fellowship} (grant \# 753750 -- RaSiR). NB would like to thank the LPTMC at Universit\'{e} Pierre et Marie Curie (Paris 06), Sorbonne Universit\'{e}s, for warm hospitality.}
\author{G.~H.~E.~Duchamp}
\address{Universit\'{e} Paris 13, Sorbonne Paris Cit\'{e}, LIPN, CNRS UMR 7030, F-93430 Villetaneuse, France}
\email{ghed@lipn.univ-paris13.fr}

\author{K.~A.~Penson}
\address{Universit\'{e} Pierre et Marie Curie (Paris 06), Sorbonne Universit\'{e}s, LPTMC,  CNRS UMR 7600, F-75252 Paris Cedex 05, France}
\email{penson@lptl.jussieu.fr}

\date{}


\maketitle

\begin{abstract}
For a sequence $P=(p_n(x))_{n=0}^{\infty}$ of polynomials $p_n(x)$, we study the $K$-tuple and $L$-shifted exponential lacunary generating functions $\mathcal{G}_{K,L}(\lambda;x):=\sum_{n=0}^{\infty}\frac{\lambda^n}{n!} p_{n\cdot K+L}(x)$, for $K=1,2\dotsc$ and $L=0,1,2\dotsc$. We establish an algorithm for efficiently computing $\mathcal{G}_{K,L}(\lambda;x)$ for generic polynomial sequences $P$. This procedure is exemplified by application to the study of Hermite polynomials, whereby we obtain closed-form expressions for $\mathcal{G}_{K,L}(\lambda;x)$ for arbitrary $K$ and $L$, in the form of infinite series involving generalized hypergeometric functions. The basis of our method is provided by certain resummation techniques, supplemented by operational formulae. Our approach also reproduces all the results previously known in the literature. 
\end{abstract}

\section{Introduction}

Lacunary generating functions appeared previously in a number of circumstances, including for example the treatment of Cauchy problems in partial differential equations~\cite{babusci2017lacunary,penson2018quasi}. Here, we develop a rather general technique for the treatment of such generating functions, applicable to sequences $P=(p_n(x,y))_{n=0}^{\infty}$ of polynomials $p_n(x,y)$, where $x$ is the generic variable and $y$ plays the role of a parameter. Such two-variable extensions of one-variable polynomials have been strongly advocated in~\cite{babusci2010lectures}. They can be logically and consistently defined for all standard families of orthogonal polynomials such as Hermite, Laguerre, Chebyshev of first and second kind, Jacobi and Legendre polynomials~\cite{babusci2017lacunary,babusci2010lectures,beals2016special}. Once such two-variable equivalents are properly defined, their one-variable variants are obtained by fixing the values of both variables to functions of one of the variables. The concrete example considered in this paper is given by the two-variable Hermite (or so-called Hermite-Kamp\'{e} de F\'{e}riet) polynomials $H_n(x,y)$~\cite{kampe,dattoli1997evolution}, from which the standard one-variable Hermite polynomials $H_n(x)$ may be recovered via (see~\eqref{eq:HPtwo} below for the definition of $H_n(x,y)$)
\begin{equation}
H_n(x)=H_n(2x,-1)\,.
\end{equation}
We will focus our particular attention onto the derivation of a general formula for the $K$-tuple $L$-shifted lacunary generating functions $\cH_{K,L}(\lambda;x,y)$ of the two-variable Hermite polynomials $H_n(x,y)$, which are defined (for $K=1,2,3,\dotsc$ and $L=0,1,2,\dotsc$) by
\begin{equation}\label{eq:HLGFa}
\cH_{K,L}(\lambda;x,y):=\sum_{n=0}^{\infty}\frac{\lambda^n}{n!}\; H_{n\cdot K+L}(x,y)\,.
\end{equation}
The exponential generating functions of type~\eqref{eq:HLGFa} for Hermite and other types polynomials are very sparsely known, and progress in obtaining new closed-form formulas has been painstakingly slow. A glance at standard reference tables~\cite{prudnikov1992integrals} reveals only a few known examples. A number of results in this vein were obtained by combinatorial approaches initiated by D.~Foata and V.~Strehl in~\cite{foataStrehl1984} supplemented by umbral methods, see~\cite{dattoli2017operational} and references therein. This methodology culminated recently in a tandem study of various lacunary generating functions of Laguerre polynomials derived by purely umbral-type~\cite{babusci2017lacunary} and purely combinatorial methods~\cite{strehl2017lacunary}. Only a few results are currently available for lacunary generating functions of Hermite polynomials: the double lacunary case has been combinatorially re-derived by D.~Foata in~\cite{foata1981some}, whereas the more challenging triple-lacunary generating function has been derived by both umbral and combinatorial methods in~\cite{gessel2005triple}. Finally, several new lacunary generating functions for Legendre and Chebyshev polynomials were obtained recently by a combination of analytic and umbral methods in~\cite{gorskalacunary}. To conclude our short survey of results known previously in the literature, let us comment that there exists a related result due to Nieto and Truax~\cite{nieto1995arbitrary}, which (in the form adapted to the two-variable Hermite polynomials $H_n(x,y)$ as presented in~\cite{dattoli1997evolution,dattoli1998operational}) reads for $K\in \bZ_{\geq 1}$, $L\in \bZ_{\geq 0}$ and $L<K$
\begin{subequations}
\begin{align}
\cS_{K,L}(\lambda;x,y)&:=\sum_{n=0}^{\infty}\frac{\lambda^{n\cdot K+L}}{(n\cdot K+L)!}H_{n\cdot K+L}(x,y)\\
&=\frac{1}{K}\sum_{\ell=1}^K\; \frac{e^{x\tau_{\ell}+y\tau_{\ell}^2}}{e^{2\pi i \ell L/K}}\,,\quad 
\tau_{\ell}:=\lambda e^{2\pi i \ell/K}\,.
\end{align}
\end{subequations}
Note however that this type of generating function is not the lacunary exponential type studied in the present paper.\\

We believe to have extended the knowledge of these objects by providing a general methodology for the study of lacunary generating functions. Indeed, our results are obtained by specialization of a general algorithm developed by us here (see Lemma~1 and Corollary~1). Quite remarkably,  it takes as its input just the coefficients $g_{r,m}(y)$ of the \emph{elementary} exponential generating functions $\cG_{1,0}(\lambda;x,y)$, defined via
\begin{subequations}
\begin{align}
\cG_{1,0}(\lambda;x,y)&=\sum_{n=0}^{\infty}\frac{\lambda^n}{n!} p_n(x,y)\\
&=\sum_{r=0}^{\infty}x^r\sum_{m=0}^{\infty}\frac{\lambda^{r+m}}{(r+m)!}\; g_{r\!,\,m}(y)\,,
\end{align}
\end{subequations}
and returns (for arbitrary parameters $K=1,2,\dotsc$ and $L=0,1,\dotsc$) closed-form expressions for the $K$-tuple $L$-shifted lacunary generating functions $\cG_{K,L}(\lambda;x,y)$. Armed with this technique, we derive \emph{all} lacunary generating functions $\cH_{K,L}(\lambda;x,y)$ for the two-variable Hermite polynomials $H_n(x,y)$ (see Theorem~\ref{thm:HPlacK}, and also Table~\ref{tab:HPlacFull} and Table~\ref{tab:HPlacShift} for a number of examples).\\

The paper is structured as follows: in Section~\ref{sec:Alg}, we derive our general algorithm for computing lacunary generating functions. These results are specialized in Section~\ref{sec:HP} to the case of two-variable Hermite-polynomials. Additional details of the proofs are provided in Appendix~\ref{app:proofs}. For the readers' convenience, an elementary Maple code for algorithmically verifying some of our explicit formulae of Table~\ref{tab:HPlacFull} is given in Appendix~\ref{app:Maple}.

\section{An algorithm for computing lacunary generating functions}
\label{sec:Alg}

Suppose we are given a sequence of polynomials $p_n(x,y)$, for $n\geq 0$, i.e.\ one polynomial for each non-negative integer index $n$ and such that 
\[
\deg_x(p_n(x,y))=n\,.
\]
Typically, we will consider the variable $y$ as a formal parameter, and focus on expansions in terms of the variable $x$. Note also that one could without additional complications admit multiple variables $y_1,y_2,\dotsc$ instead of a single ``parameter'' variable $y$, so we will, without loss of generality, consider here only the case of a single ``parameter'' variable $y$.\\

Given a particular set of polynomials, one may compute the exponential generating function $\cG(\lambda;x,y)$ for this set as
\begin{subequations}
\begin{align}
\cG(\lambda;x,y)&:=\sum_{n=0}^{\infty}\frac{\lambda^n}{n!} \; p_n(x,y)\label{eq:EGFa}\\
&=\sum_{r=0}^{\infty} x^r \sum_{m=0}^{\infty} \frac{\lambda^{r+m}}{(r+m)!}\; g_{r\!,\,m}(y)\,. \label{eq:EGFb}
\end{align}
\end{subequations}
While the form as presented in~\eqref{eq:EGFa} merely amounts to re-encoding of the information available via the explicit definition of the polynomials $p_n(x,y)$, the form~\eqref{eq:EGFb} in fact necessitates a calculation: equation~\eqref{eq:EGFa} must be expanded into powers of $x$ and then further in powers of the formal variable $\lambda$, which (in the specific pairing of powers $\lambda^{r+m}$ with $1/(r+m)!$) defines a set of expansion coefficients $g_{r\!,\,m}(y)$. To provide a concrete example, we demonstrate the calculation for the case of the two-variable Hermite polynomials $H_n(x,y)$ that will play a central role later. Their exponential generating function $\cH(\lambda;x,y)$ reads~\cite{babusci2010lectures,dattoli1997evolution}
\begin{equation}\label{eq:EGFhp1}
\cH(\lambda;x,y):=\sum_{n=0}^{\infty}\frac{\lambda^n}{n!} H_n(x,y)=e^{\lambda x+\lambda^2 y}\,.
\end{equation}
To obtain the form as described in~\eqref{eq:EGFb}, we expand the EGF first in terms of the variable $x$, and then in terms of the variable $\lambda$:
\begin{equation}
\begin{aligned}
\cH(\lambda;x,y)&=e^{\lambda x+\lambda^2 y}\\
&=\left(\sum_{r=0}^{\infty} x^r\; \tfrac{\lambda^r}{r!}\right)\;e^{\lambda^2 y}\\
&=\sum_{r=0}^{\infty} x^r \; \sum_{m=0}^{\infty} \frac{\lambda^{r+2m}}{r! m!}y^m\\
&=\sum_{r=0}^{\infty} x^r \; \sum_{m=0}^{\infty} \frac{\lambda^{r+2m}}{(r+2m)!}\; \left(\frac{(r+2m)!\; y^m}{r!m!}\right)\\
&=\sum_{r=0}^{\infty}x^r \sum_{m=0}^{\infty}\frac{\lambda^{r+2m}}{(r+2m)!}\; h_{r\!,\,2m}(y)\,.
\end{aligned}
\end{equation}
Consequently, for these polynomials by comparison with the defining equation~\eqref{eq:EGFb}, we find that $g_{r\!,\,m}(y)=0$ for $m$ odd, and 
\begin{equation}
g_{r\!,\,2m}(y)=\frac{(r+2m)! y^m}{r!m!}\,.
\end{equation}

Given again a set of polynomials $p_n(x,y)$, one may also define a more general family of generating functions, the so-called \emph{$K$-tuple $L$-shifted lacunary generating functions} $\cG_{K,L}(\lambda;x,y)$ (for $K\in \bZ_{\geq 1}$ and $L\in \bZ_{\geq 0}$):
\begin{subequations}
\begin{align}
\cG_{K,L}(\lambda;x,y)&:=
\sum_{n=0}^{\infty}\frac{\lambda^n}{n!} \; 
p_{K\cdot n+L}(x,y)\label{eq:LGFa}\\
&=\sum_{r=0}^{\infty} x^r \sum_{m=0}^{\infty} \frac{\lambda^{r+m}}{(r+m)!}\; g^{(K,L)}_{r\!,\,m}(y)\,. \label{eq:LGFb}
\end{align}
\end{subequations}
Note that by definition,
\begin{equation}\label{eq:LGFegf}
\cG_{1,0}(\lambda;x,y)\equiv\cG(\lambda;x,y)\,.
\end{equation}

Before we continue with further elaborations, we wish to clarify the precise framework within which the series rearrangements that are essential to our methods are well-posed.

\subsection{Summability}\label{sec:sum}

Polynomials or series of any sort are functions $M\to k$ where $M$ is a set of monomials (a monoid\footnote{The product of two monomials is a monomial, there exists a neutral for this multiplication, the void monomial.}) and $k$ its set of coefficients ($\mathbb{R},\mathbb{C}$ any ring as a ring of polynomials or series). In order to manipulate safely infinite sums, such as in the transition from~\eqref{eq:EGFa} to the form~\eqref{eq:EGFb}, these spaces are endowed with the notion of \emph{summability}. 
\begin{defn} A family of series $(S_i)_{i\in I}$ is said {\it summable} \cite{reutenauer} if it is {\it locally finite} \cite{berstel} i.e. if 
\begin{equation}
(\forall m\in M)\big(\; \vert\{i\in I|\scal{S_i}{m}\not=0\}\vert<\infty\;\big)\,,
\end{equation}
where $\scal{S_i}{m}$ stands for ``the coefficient of the monomial $m$ in $S_i$''.\\ 
In this case ($(S_i)_{i\in I}$ is summable) we say that $S=\sum_{i\in I}S_i$ where, for all $m\in M$
\begin{equation}
\scal{S}{m}=\sum_{i\in I}\scal{S_i}{m}\,.
\end{equation}
\end{defn}  
\begin{remarks} i) In $k[|x]]$ (series of one variable) $M=\{x^n\}_{n\geq 0}$, the family $\{(x+x^2)^n\}_{n\geq 0}$  is summable (with sum $\frac{1}{1-x-x^2}$), while $\{(x+1)^n\}_{n\geq 0}$ is not.\\
ii) As an ``a posteriori'' justification, for all series $S$, one can easily check that the family 
$(\scal{S}{m}\,m)_{m\in M}$ is summable with sum $S$, hence the additive notation can be rigorously set as 
\begin{equation}
S=\sum_{m\in M} \scal{S}{m}\,m\,.
\end{equation}  
iii) The reader aware of topology will find that this notion is none other that the {\rm summability for the topology of point-wise convergence (on $M$), $k$ being endowed with the discrete topology}~\cite{GT1}.\\
iv) The sequence $\Big((1+\frac{x}{n})^{n}\Big)_{n\geq 1}$ converges towards $e^x$, but coefficient by coefficient, this is a coarser topology, with -- this time -- $\mathbb{R}$ endowed with the usual topology, i.e.\ Treves' topology~\cite{SMF,treves}.          
\end{remarks}

We will now present first a closed-form equation that makes explicit the relationship between the expansion coefficients $g_{r,\,m}^{(K,L)}(y)$, see equation~\eqref{eq:LGFb}, and the expansion coefficients $g^{(1,0)}_{r,\,m}(y)$ of the exponential generating function $\cG_{1,0}(\lambda;x,y)$, see equation~\eqref{eq:LGFegf}. 

\subsection{Lacunary shifts}

It follows directly from the definition given in~\eqref{eq:LGFa} that the operator
\begin{equation}
\bS_L:=\left(\frac{\partial}{\partial \lambda}\right)^L
\end{equation}
implements the action of lacunary shifts on exponential generating functions $\cG_{K,0}(\lambda;x,y)$:
\begin{equation}\label{eq:lacShift}
\bS_L \left(\cG_{K,0}(\lambda;x,y)\right)=\sum_{n=L}^{\infty}\frac{\lambda^{n-L}}{(n-L)!}\; p_{n\cdot K}(x,y)
=\sum_{n=0}^{\infty}\frac{\lambda^{n}}{n!}\; p_{n K+L}(x,y)=\cG_{K,L}(\lambda;x,y)\,.
\end{equation}
While it may be intricate in its own right to extract from this equation the lacunary expansion coefficients $g^{(1,L)}_{r,\,m}(y)$ as defined in~\eqref{eq:LGFb}, this computation is in principle straightforward. Therefore we will focus in this paper mostly on providing formulae for the lacunary dilatations characterized by an integer $K$. However, for the explicit application example as presented in Section~\ref{sec:HP}, we will be able to take advantage of certain operational techniques particular to the Hermite polynomials in order to compute the lacunary shifts of their respective generating functions in closed form.

\subsection{Lacunary dilatations}

According to the definition given in~\eqref{eq:LGFa}, the so-called \emph{$K$-fold lacunary dilatation} (for $K\in \bZ_{\geq 1}$) is implemented by means of the (formal) operator, which acts on formal series $F(\lambda)$ via
\begin{equation}\label{def:lacD}
\bL_K(F(\lambda)):=F(\lambda)\big\vert_{
\lambda^n\mapsto 
\; \delta_{(n \Mod K),0}
\; \frac{n!}{\left(n/K\right)!}\;\lambda^{\left(n/K\right)}}\,.
\end{equation}
Consequently, the action of $\bL_K$ may be considered as a transformation of monomials in the spirit of Section~\ref{sec:sum}. Indeed,
\begin{subequations}\label{eq:lacDil}
\begin{align}
\bL_K\left(\cG_{1,L}(\lambda;x,y)\right)&=\sum_{n= 0}^{\infty}\frac{\lambda^n}{n!}\;\left(\delta_{(n \Mod K),0}\;\frac{n!}{(n/K)!}\;\frac{\lambda^{(n/K)}}{\lambda^n}\right)\; p_{n+L}(x,y)\\
&=\sum_{r=0}^{\infty}\frac{\lambda ^r}{r!} p_{r\cdot K+L}(x,y)=\cG_{K,L}(\lambda;x,y)\,,\label{eq:auxKa}
\end{align}
\end{subequations}
where in~\eqref{eq:auxKa}, $r=n/K$.\\

Note at this point that the operators $\bS_L$ and $\bL_K$ (for $K=1,2,\dotsc$ and $L=0,1,\dotsc$) allow to implement the derivation of arbitrary lacunary generating functions $\cG_{K,L}(\lambda;x,y)$ from the initial knowledge of the exponential generating function $\cG_{1,0}(\lambda;x,y)$ as defined in~\eqref{eq:LGFegf}.\\

Assume now that we are given either the exponential generating function $\cG_{1,0}(\lambda;x,y)$ of some polynomials $p_n(x,y)$, or the $L$-shifted version $\cG_{1,L}(\lambda;x,y)$ thereof, and introduce for notational simplicity the notation $\cG(\lambda;x,y)$ for either choice of generating function (in a slight abuse in light of previous notations). We thus assume that we are given an equation of the generic form 
\begin{equation}\label{eq:GFexcpansion}
\cG(\lambda;x,y)=\sum_{r=0}^{\infty}x^r\sum_{m=0}^{\infty}\frac{\lambda^{r+m}}{(r+m)!}\; g_{r\!,\,m}(y)\,.
\end{equation}
With these notational preparations, computing a $K$-fold lacunary dilatation of $\cG(\lambda;x,y)$ amounts to extracting the coefficients $g^{(K,0)}_{r\!,\,m}(y)$ in explicit form via expanding $\bL_K(\cG(\lambda;x,y))$ in powers of $x$, and then further in powers of $\lambda$. The only technical tool necessary is an elementary procedure to resolve the summations over $\Mod K$-classes involved in the definition.

\begin{restatable}{lem}{genKD}\label{lem:KD}
The explicit form for the $K$-fold lacunary dilatation $\bL_K(\cG(\lambda;x,y))$ ($K\in \bZ_{\geq 1}$)  of a generating function $\cG(\lambda;x,y)$ with expansion coefficients $g_{r,m}(y)$ reads 
\begin{equation}\label{eq:KLGF}
\begin{aligned}
\cG_{K,0}(\lambda;x,y)&=\bL_K\left(\cG(\lambda;x,y)\right)\\
&= 
\sum_{s=0}^{\infty} x^{s\cdot K}\sum_{q=0}^{\infty}\frac{\lambda^{s+q}}{(s+q)!}\;g_{s\cdot K,\,q\cdot K}(y)\\
&\quad 
+\sum_{\alpha=1}^{K-1}\sum_{s=0}^{\infty} x^{(s+1)\cdot K-\alpha}\sum_{q=0}^{\infty}\frac{\lambda^{s+q+1}}{(s+q+1)!}\; g_{(s+1)\cdot K-\alpha,\,q\cdot K+\alpha}\,.
\end{aligned}
\end{equation}
\end{restatable}
\begin{proof}
The argument follows from splitting the summation over $r$ in $\bL_K(\cG(\lambda;x,y)$ (with $\cG(\lambda;x,y)$ expanded in the form described in~\eqref{eq:EGFb}) into $\Mod K$-classes, which consequently leads to a modification of the summation over the second index $m$. We refer the readers to Appendix~\ref{app:lemP} for the precise details. 
\end{proof}

In practical applications, it is often the case that the expansion coefficients $g_{r,m}$ have further structure, e.g.\ as in the case of the Hermite polynomials (discussed in the next section) that only the coefficients with second index $m$ even are non-zero. We thus present the following corollary for this specific scenario, which will play an important role in the explicit calculation of closed-form expressions for the lacunary generating functions of the Hermite polynomials. Note also that one could easily adapt our approach to other scenarios in which the second indices $m$ fulfill a $\Mod N$ constraint (for $N\in \bZ_{\geq 3}$ some non-negative integer) via an entirely analogous argument.

\begin{restatable}{cor}{genKDcor}\label{cor:KD}
The explicit form for the $K$-fold lacunary shift $\bL_K(\cG(\lambda;x,y))$ ($K\in \bZ_{\geq 1}$)  of a generating function $\cG(\lambda;x,y)$ with expansion coefficients $g_{r\!,\,m}(y)$ splits into a part which only contains summands with \emph{even} second indices $m$ of $g_{r,m}(y)$ (marked ``part $E$'') and a part involving odd second indices only (marked ``part $O$''). It reads for the case $K=2T$ (with $T\in \bZ_{\geq 1}$)
\begin{subequations}\label{eq:KLGFeven}
\begin{align}
\bL_{K=2T}\left(\cG(\lambda;x,y)\right)&= {\bigg[}^{\text{part $E$}}
\sum_{s=0}^{\infty} x^{s\cdot K}\sum_{q=0}^{\infty}
\frac{\lambda^{s+q}}{(s+q)!}\; g_{s\cdot K,\,q\cdot K}(y)\\
&\quad 
+\sum_{\beta=1}^{T-1}\sum_{s=0}^{\infty} x^{(s+1)\cdot K-2\beta}\sum_{q=0}^{\infty}\frac{\lambda^{s+q+1}}{(s+q+1)!}\; g_{(s+1)\cdot K-2\beta,\,q\cdot K+2\beta}\bigg]\\
&\quad+{\bigg[}^{\text{part $O$}}
\sum_{\beta=1}^T\sum_{s=0}^{\infty}x^{(s+1)\cdot K-2\beta+1}\sum_{q=0}^{\infty}\frac{\lambda^{s+1+1}}{(s+q+1)!}\; g_{(s+1)\cdot K-2\beta+1,\,q\cdot K+2\beta-1}\bigg]\,,
\end{align}
\end{subequations}
and for the case $K=2T+1$ (with $T\in \bZ_{\geq 1}$)
\begin{subequations}\label{eq:KLGFodd}
\begin{align}
\bL_{K=2T+1}\left(\cG(\lambda;x,y)\right)&= {\bigg[}^{\text{part $E$}}\sum_{s=0}^{\infty} x^{s\cdot K}\sum_{\ell=0}^{\infty}
\frac{\lambda^{s+2\ell}}{(s+2\ell)!}\; g_{s\cdot K,\,2\ell \cdot K}(y)\\
&\quad 
+\sum_{\beta=1}^{T}\sum_{s=0}^{\infty} x^{(s+1)\cdot K-2\beta}\sum_{\ell=0}^{\infty}\frac{\lambda^{s+2\ell+1}}{(s+2\ell+1)!}\; g_{(s+1)\cdot K-2\beta,\,2\ell\cdot K+2\beta}\\
&\quad 
+\sum_{\beta=1}^{T}\sum_{s=0}^{\infty} x^{(s+1)\cdot K-2\beta+1}\sum_{\ell=0}^{\infty}\frac{\lambda^{s+2\ell+2}}{(s+2\ell+2)!}\; g_{(s+1)\cdot K-2\beta+1,\,(2\ell+1)\cdot K+2\beta-1}\bigg]\\
&\quad+{\bigg[}^{\text{part $O$}}
\sum_{s=0}^{\infty} x^{s\cdot K}\sum_{\ell=0}^{\infty}
\frac{\lambda^{s+2\ell+1}}{(s+2\ell+1)!}\; g_{s\cdot K,\,(2\ell+1) \cdot K}(y)\\
&\quad 
+\sum_{\beta=1}^{T}\sum_{s=0}^{\infty} x^{(s+1)\cdot K-2\beta+1}
\sum_{\ell=0}^{\infty}\frac{\lambda^{s+2\ell+1}}{(s+2\ell+1)!}\; g_{(s+1)\cdot K-2\beta+1,\,(2\ell)\cdot K+2\beta-1}\\
&\quad 
+\sum_{\beta=1}^{T}\sum_{s=0}^{\infty} x^{(s+1)\cdot K-2\beta}\sum_{\ell=0}^{\infty}\frac{\lambda^{s+2\ell+2}}{(s+2\ell+2)!}\; g_{(s+1)\cdot K-2\beta,\,(2\ell+1)\cdot K+2\beta}\bigg]\,.
\end{align}
\end{subequations}
\end{restatable}
\begin{proof}
The proof follows from a subdivision of summation ranges as detailed in Appendix~\ref{app:corDproof}.
\end{proof}

\section{Hermite polynomials}
\label{sec:HP}

Due to their paramount importance in special functions theory and in the development of operational techniques (cf.\ e.g.~\cite{dattoli1997evolution,dattoli2000generalized}), we will now present as an application of our general algorithm the computation of all higher order lacunary generating functions for the \emph{Hermite-Kamp\'{e} de F\'{e}riet polynomials}~\cite{kampe,dattoli1997evolution} (with $n\in \bZ_{\geq0}$):
\begin{equation}\label{eq:HPtwo}
H_n(x,y):=e^{y\frac{d^2}{dx^2}} x^n =n! \sum_{k=0}^{\lfloor \frac{n}{2}\rfloor}\frac{x^{n-2k}y^k}{(n-2k)!k!}\,.
\end{equation}
It is well-known that these polynomials possess the exponential generating function (compare~\eqref{eq:EGFhp1})
\begin{subequations}
\begin{align}
\cH_{1,0}(\lambda;x,y)&:=\sum_{n=0}^{\infty}
\frac{\lambda^n}{n!}
H_n(x,y)=e^{y\frac{d^2}{dx^2}} e^{\lambda x}\label{eq:HPlac1s0}\\
&=e^{\lambda x + \lambda^2 y}\,.
\end{align}
\end{subequations}
Here and from now on, we will take the notational convention that the generating functions $\cG_{K,L}(\lambda;x,y)$ for the two-variable Hermite polynomials $H_n(x,y)$ will be denoted $\cH_{K,L}(\lambda;x,y)$.\\

The above formula is a straightforward consequence of the Crofton identity~\cite{crofton79} (see also~\cite[Eq.~(I.3.17)]{dattoli1997evolution} and~\cite{dattoli2000generalized}),  according to which, for $m\in \bZ_{\geq 1}$, $\lambda$ a formal variable and $f(x),g(x)$ formal power series in $x$, we have the following equality of operators (which are the expressions between $[\dotsc]$):
\begin{equation}\label{eq:crofton}
\left[\exp\left(\lambda \frac{d^m}{dx^m}\right)f(x)\right]g(x)=\left[ f\left(x + m\lambda \frac{d^{m-1}}{dx^{m-1}}\right)\exp\left(\lambda\frac{d^m}{dx^m}\right)\right]g(x)\,.
\end{equation}
As a second important consequence of this identity, we find a significant reduction of the complexity of determining the lacunary generating functions of the Hermite polynomials: surprisingly enough, instead of having to compute explicitly the shifts $\cH_{1,L}(\lambda;x,y)$ of the exponential generating function $\cH_{1,0}(\lambda;x,y)$, followed by application of Corollary~\ref{cor:KD} to compute $\cH_{K,L}(\lambda;x,y)$, we can take advantage of operational techniques to derive an exponential generating for \emph{all} shifts in closed form. Note that for some ranges of parameters (i.e.\ for $K=1,2$ and $L\in \bZ_{\geq 1}$), formulae equivalent to our result~\eqref{eq:propHLGF} below have been presented already in~~\cite{dattoli1997evolution,dattoli1998operational}.

\begin{prop}\label{prop:LshiftHLGF}
Define the exponential generating function\footnote{In a sense, $\cR_K(\mu;\lambda;x,y)$ is the exponential generating function of other types of exponential generating functions, a slightly unusual, yet rather useful construct. Note also that these manipulations are admissible in the framework according to Section~\ref{sec:sum}.} of the lacunary shifts $\cH_{K,L}(\lambda;x,y)$ of the $K$-tuple lacunary generating function $\cH_{K,0}(\lambda;x,y)$ of the polynomials $H_n(x,y)$ as
\begin{equation}
\begin{aligned}
\cR_K(\mu;\lambda;x,y)&:= \sum_{L=0}^{\infty}\frac{\mu^L}{L!} \cH_{K,L}(\lambda;x,y)\\
&\overset{\eqref{eq:lacDil}}{=}\sum_{L=0}^{\infty}\frac{\mu^L}{L!} \bL_K\left(\cH_{1,L}(\lambda;x,y)\right)\\
&\overset{\eqref{eq:lacShift}}{=}\sum_{L=0}^{\infty}\frac{\mu^L}{L!} \bL_K\left(\left(\tfrac{\partial}{\partial\lambda}\right)^L\left(\cH_{1,0}(\lambda;x,y)\right)\right)\\
&=\bL_K\left(e^{\mu \frac{\partial}{\partial \lambda}}\cH_{1,0}(\lambda;x,y)\right)\,.
\end{aligned}
\end{equation}
Then by virtue of the Crofton identity and of a semi-linear normal-ordering technique~\cite{dattoli1997evolution}, one finds that
\begin{equation}\label{eq:propHLGF}
\boxed{\cR_K(\mu;\lambda;x,y)=e^{\mu x+\mu^2 y}\cH_{K,0}(\lambda;x+2\mu y,y)=\cH_{1,0}(\mu;x,y)\cH_{K,0}(\lambda;x+2\mu y,y)\,.}
\end{equation}
Consequently, as one of the key results of this paper, we find that
\begin{equation}\label{eq:HKLexpl}
\cH_{K,L}(\lambda;x,y)=\left[\left(\tfrac{\partial}{\partial\mu}\right)^L \cR_K(\mu;\lambda;x,y)\right]\bigg\vert_{\mu
\mapsto0}
\end{equation}
\begin{proof}
Partially evaluating the formula for $\cH_K(\mu;\lambda;x,y)$,
\begin{subequations}
\begin{align}
\cR_K(\mu;\lambda;x,y)&=\bL_K\left(e^{\mu \frac{\partial}{\partial \lambda}}\cH_{1,0}(\lambda;x,y)\right)\\
&\overset{\eqref{eq:HPlac1s0}}{=} 
\bL_K\left(e^{\mu \frac{\partial}{\partial \lambda}}e^{y\frac{d^2}{dx^2}}e^{\lambda x}\right)\\
&=\bL_K\left(e^{y\frac{d^2}{dx^2}}e^{\mu \frac{\partial}{\partial \lambda}}e^{\lambda x}\right)\\
&=\bL_K\left(e^{y\frac{d^2}{dx^2}}e^{\mu x}e^{\lambda x}\right)\,,\label{eq:shiftPropA}
\end{align}
\end{subequations}
and applying the Crofton identity~\eqref{eq:crofton} (to move the operator $\exp(y\frac{d^2}{dx^2})$ past the operator $\exp(\mu x)$) yields
\begin{equation}
\cR_K(\mu;\lambda;x,y)=\bL_K\left(e^{\mu(x+2y\frac{d}{dx})}e^{y\frac{d^2}{dx^2}}e^{\lambda x}\right)\,.
\end{equation}
At this point, note that the operator $\exp(\mu(x+2y\tfrac{d}{dx}))$ is independent of the formal variable $\lambda$, whence this operator and the lacunary dilatation operator $\bL_K$ commute:
\begin{subequations}
\begin{align}
\cR_K(\mu;\lambda;x,y)&=e^{\mu(x+2y\frac{d}{dx})}\left(\bL_K\left(e^{y\frac{d^2}{dx^2}}e^{\lambda x}\right)\right)\\
&\overset{\eqref{eq:HPlac1s0}}{=}
e^{\mu(x+2y\frac{d}{dx})}\left(\bL_K\left(\cH_{1,0}(\lambda;x,y)\right)\right)\\
&\overset{\eqref{eq:lacDil}}{=}e^{\mu(x+2y\frac{d}{dx})} \cH_{K,0}(\lambda;x,y)\,.
\end{align}
\end{subequations}
This formula may be further improved via applying a certain semi-linear normal ordering technique (see Appendix~\ref{app:SLNO} for the precise details) to obtain the final result as stated, whence for every formal power series $f(x)$ (and for $\mu$ a formal variable),
\begin{equation}
\begin{aligned}
e^{\mu(x+2y\frac{d}{dx})}f(x)&=g(\mu;x)f(T(\mu;x))\\
g(\mu;x)&=e^{\mu x+\mu^2 y}\,,\quad T(\mu;x)=x+2\mu y\,.
\end{aligned}
\end{equation}
\end{proof}
\end{prop}

Due to this proposition, our task of determining closed-form equations for all higher order lacunary generating functions $\cH_{K,L}(\lambda;x,y)$ effectively collapses to the slightly simpler task of determining the lacunary generating functions $\cH_{K,0}(\lambda;x,y)$ for $K\in \bZ_{\geq 1}$. We will employ the generic algorithm as presented in Lemma~\ref{lem:KD} in the variant as presented in Corollary~\ref{cor:KD}, since the expansion coefficients $h^{(1,0)}_{r\!,\,m}(y)$ of the exponential generating function $\cH_{1,0}(\lambda;x,y)$ defined via
\begin{equation}
\cH_{1,0}(\lambda;x,y):=\sum_{n=0}^{\infty}\frac{\lambda^n}{n!} H_n(x,y)=e^{\lambda x+\lambda^2 y}
=\sum_{r=0}^{\infty}x^r \sum_{m=0}^{\infty}\frac{\lambda^{r+m}}{(r+m)!}\; h^{(1,0)}_{r\!,\,m}(y)
\end{equation}
fulfill
\begin{equation}\label{eq:HPcoeffs}
h^{(1,0)}_{r\!,\,m}(y)=\delta_{(m\Mod 2),0}\; \frac{(r+m)!y^{\frac{m}{2}}}{r!(\frac{m}{2})!}\,.
\end{equation}
This entails in particular that our task further reduces to applying the explicit formula for the ``even parts'' (marked ``part $E$'') given in Corollary~\ref{cor:KD} for $K=2T$ and $K=2T+1$, respectively.\\

It will further prove convenient to recall the \emph{Gauss-Legendre multiplication formula}~\cite[Eq.~5.5.6]{NIST:DLMF} for Gamma functions (for $n\cdot z\not\in \bZ_{\leq 0}$),
\begin{equation}\label{eq:GMF}
\Gamma(nz)=n^{nz-\frac{1}{2}}(2\pi)^{\frac{(1-n)}{2}}\prod_{j=0}^{n-1}\Gamma\left(z+\frac{j}{n}\right)\,.
\end{equation} 
More precisely, we will make use of the following variant of this formula: for $n(s+x),nx\in \bC\setminus \bZ_{\leq 0}$, $n\in \bZ_{\geq 2}$ and $s\in \bZ_{\geq 0}$, we have that
\begin{equation}\label{eq:GMFC}
\Gamma(n(s+x))=\left(n^{s\cdot n}\right)\Gamma(nx)\;\prod_{j=0}^{n-1}\left(x+\frac{j}{n}\right)_s\,,
\end{equation}
with the Pochhammer symbol $(a)_b$ defined according to the convention
\[
(a)_b:=\frac{\Gamma(a+b)}{\Gamma(a)}\,.
\]

Let us also fix notations for hypergeometric functions~\cite{prudnikov1992integrals,beals2016special} as
\begin{equation}
\pFq{p}{q}{a_1,\dotsc,a_p}{b_1,\dotsc,b_q}{z}:=\sum_{s=0}^{\infty}\frac{z^s}{s!}\frac{(a_1)_s\dotsc (a_p)_s}{(b_1)_s\dotsc (b_q)_s}\,,
\end{equation}
and let us denote sequences in the form $\seq{f_i}{1\leq i\leq n}$, such as in
\begin{equation}
\pFq{p}{q}{\seq{a_i}{1\leq i\leq p}}{\seq{b_j}{1\leq i\leq q}}{z}:=\pFq{p}{q}{a_1,\dotsc,a_p}{b_1,\dotsc,b_q}{z}\,.
\end{equation}

With these preparations, we are now in a position to state another key result of this paper, whence a closed-form expression for all $K$-tuple lacunary generating functions $\cH_{K,0}(\lambda;x,y)$ of the Hermite polynomials $H_n(x,y)$ in infinite-series form:

\begin{restatable}{thm}{HPlacK}\label{thm:HPlacK}
The $K$-tuple lacunary generating functions $\cH_{K,0}(\lambda;x,y)$ read for $K=2T$ ($T\in \bZ_{\geq 1}$)
\begin{equation}
\begin{aligned}\label{eq:lacHPeven}
\cH_{K=2T,0}(\lambda;x,y)&=
\sum_{s=0}^{\infty}\tfrac{\lambda^s}{s!}\;x^{2T s}\;
\pFq{(2T-1)}{(T-1)}{%
\seq{s+\tfrac{j+1}{2T}}{0\leq j\leq 2T-2}}{%
\seq{\tfrac{\ell+1}{T}}{0\leq \ell\leq T-2}}{\lambda(4Ty)^T}\\
&\quad 
+\sum_{\beta=1}^{T-1}\sum_{s=0}^{\infty}\tfrac{\lambda^{s+1}}{(s+1)!}\; 
x^{2T (s+1)-2\beta}\; y^{\beta}\;\left(\tfrac{(2T(s+1))!}{(2T(s+1)-2\beta)!\beta!}\right)\cdot\\
&\qquad\cdot
\pFq{(2T-1)}{(T-1)}{\seq{s+1+\tfrac{j+1}{2T}}{0\leq j\leq 2T-2}}{%
\seq{\tfrac{\beta+\ell+1}{T}}{\stackrel{0\leq\ell\leq T-1}{\ell\neq T-1-\beta}}}{\lambda(4Ty)^T}\,,
\end{aligned}
\end{equation}
and for $K=2T+1$ (with $T\in \bZ_{\geq1}$)
\begin{equation}
\begin{aligned}\label{eq:lacHPodd}
\cH_{K=2T+1,0}(\lambda;x,y)&=
\sum_{s=0}^{\infty}\tfrac{\lambda^s}{s!}\;x^{K s}\;
\pFq{(2K-2)}{(K-1)}{\seq{\tfrac{s}{2}+\tfrac{j+1}{2K}}{\stackrel{0\leq j\leq 2K-2}{j\neq K-1}}}{\seq{\tfrac{\ell+1}{K}}{0\leq \ell \leq K-2}}{\tfrac{\lambda^2(4Ky)^K}{4}}\\
&\quad 
+\sum_{\beta=1}^{T}\sum_{s=0}^{\infty}\tfrac{\lambda^{s+1}}{(s+1)!}\; 
x^{K(s+1)-2\beta}\; y^{\beta}\;\left(\tfrac{(K(s+1))!}{(K(s+1)-2\beta)!\beta!}\right)\cdot\\
&\qquad\cdot
\pFq{(2K-2)}{(K-1)}{\seq{\tfrac{s+1}{2}+\tfrac{j+1}{2K}}{\stackrel{0\leq j\leq 2K-2}{j\neq K-1}}}{\seq{\tfrac{\beta+\ell+1}{K}}{\stackrel{0\leq \ell\leq K-1}{\ell\neq K-1-\beta}}}{\tfrac{\lambda^2(4Ky)^K}{4}}\\
&\quad 
+\sum_{\beta=1}^{T}
\sum_{s=0}^{\infty}\tfrac{\lambda^{s+2}}{(s+2)!}\; 
x^{K(s+2)-2(T+\beta)}\; y^{T+\beta}\;
\left(\tfrac{(K(s+2))!}{(K(s+2)-2(T+\beta))!(T+\beta)!}\right)\cdot\\
&\qquad\cdot
\pFq{(2K-2)}{(K-1)}{\seq{\tfrac{s+2}{2}+\tfrac{j+1}{2K}}{\stackrel{0\leq j\leq 2K-2}{j\neq K-1}}}{
\seq{\tfrac{T+\beta+\ell+1}{K}}{\stackrel{0\leq \ell\leq K-1}{\ell\neq T-\beta}}}{\tfrac{\lambda^2(4Ky)^K}{4}}\,.
\end{aligned}
\end{equation}
\end{restatable}
\begin{proof}
The proof follows by specializing Corollary~\ref{cor:KD} to the case at hand and by making repeated use of equation~\eqref{eq:GMFC} in order to transform the resulting terms into the form presented in the theorem.
\end{proof}

Since the precise details of the above calculations might be of independent interest to some of the readers, we present the case $K=4$ in the form of a worked example. Based on formula~\eqref{eq:HPcoeffs} for the expansion coefficients $h^{(1,0)}_{r\!,\,m}$ of the exponential generating function $\cH_{1,0}(\lambda;x,y)$ of the Hermite polynomials $H_n(x,y)$, and taking advantage of Corollary~\ref{cor:KD} (whence making use of the formula for even $K=2T$, of which moreover only the parts marked $part~E$ contribute in the case at hand), we obtain the following expression:
\begin{equation}
\begin{aligned}\label{eq:H40step1}
\cH_{4,0}(\lambda;x,y)&=\sum_{s=0}^{\infty} x^{4s}\sum_{q=0}^{\infty}\frac{\lambda^{s+q}}{(s+q)!}\; h^{(1,0)}_{4s\!,\,4q}(y)\\
&\quad+\sum_{s=0}^{\infty} x^{4s+2}\sum_{q=0}^{\infty}\frac{\lambda^{s+q+1}}{(s+q+1)!}\; h^{(1,0)}_{4s+2\!,\,4q+2}(y)\\
&=\sum_{s=0}^{\infty} x^{4s}\sum_{q=0}^{\infty}\frac{\lambda^{s+q}(4(s+q))!y^{2q}}{(s+q)!(4s)!(2q)!}\\
&\quad+\sum_{s=0}^{\infty} x^{4s+2}\sum_{q=0}^{\infty}\frac{\lambda^{s+q+1}(4(s+q+1))!y^{2q+1}}{(s+q+1)!(4s+2)!(2q+1)!}\,.
\end{aligned}
\end{equation}
Applying formula~\eqref{eq:GMFC} repeatedly, we find the auxiliary relations
\begin{align*}
(4(s+q))!&=4^{4q}\Gamma(4s+1)\prod_{j=0}^3\left(s+\tfrac{j+1}{4}\right)_q=4^{4q}(4s)!\left(s+\tfrac{3+1}{4}\right)_q \prod_{j=0}^{{\color{blue}2}}\left(s+\tfrac{j+1}{4}\right)_q\\
&=\tfrac{4^{4q}(4s)!(s+q)!}{q!} \prod_{j=0}^{2}\left(s+\tfrac{j+1}{4}\right)_q \qquad(\text{since }\left(s+1\right)_q=\tfrac{\Gamma(s+q+1)}{\Gamma(s+1)})\\
(2q)!
&=2^{2q}\Gamma(1)\prod_{j=0}^1\left(\tfrac{j+1}{2}\right)_q
=4^q \left(\tfrac{1}{2}\right)_q q! \qquad(\text{since }\left(1\right)_q
=\tfrac{\Gamma(q+1)}{\Gamma(1)}=q!)\,,\\
\end{align*}%
and analogously for $(4(s+q+1))!$ and $(2q+1)!$, respectively. Inserting these results into~\eqref{eq:H40step1}, we recover the formula for $\cH_{4,0}(\lambda;x,y)$ as presented in Theorem~\ref{thm:HPlacK} (see also Table~\ref{tab:HPlacFull}):
\begin{equation}
\begin{aligned}\label{eq:H40step2}
\cH_{4,0}(\lambda;x,y)&=
\sum_{s=0}^{\infty} x^{4s}\sum_{q=0}^{\infty}\tfrac{\lambda^{s+q}\;4^{4q}(4s)!(s+q)!y^{2q}\prod_{j=0}^{2}\left(s+\frac{j+1}{4}\right)_q}{(s+q)!(4s)!s! 4^q \left(\frac{1}{2}\right)_q q!}\\
&\quad
+\sum_{s=0}^{\infty} x^{4s+2}\sum_{q=0}^{\infty}\tfrac{\lambda^{s+q+1}\;4^{4q} (4(s+1))!(s+q+1)! y^{2q+1}\prod_{j=0}^{2}\left(s+1+\frac{j+1}{4}\right)_q}{(s+q+1)!(4s+2)!(s+1)!4^q q! \left(\frac{3}{2}\right)_q}\\
&=
\sum_{s=0}^{\infty} \tfrac{\lambda^s}{s!}\; x^{4s}\sum_{q=0}^{\infty}\tfrac{(\lambda (2^3y)^2)^{q}}{q!}\;\tfrac{\prod_{j=0}^{2}\left(s+\frac{j+1}{4}\right)_q}{ \left(\frac{1}{2}\right)_q}\\
&\quad
+\sum_{s=0}^{\infty} \tfrac{\lambda^{s+1}}{(s+1)!}\;x^{4s+2}y\;\left(\tfrac{(4(s+1))!}{(4(s+1)-2)!}\right)
\sum_{q=0}^{\infty}\tfrac{(\lambda (2^3y)^2)^{q}}{q!}\;\tfrac{\prod_{j=0}^{2}\left(s+1+\frac{j+1}{4}\right)_q}{ \left(\frac{3}{2}\right)_q}\\
&=
\sum_{s=0}^{\infty} \tfrac{\lambda^s}{s!}\; x^{4s}
\pFq{3}{1}{s+\tfrac{1}{4},s+\tfrac{2}{4},s+\tfrac{3}{4}}{\tfrac{1}{2}}{\lambda y^22^6}\\
&\quad
+\sum_{s=0}^{\infty} \tfrac{\lambda^{s+1}}{(s+1)!}\;x^{4s+2}y\;\left(\tfrac{(4(s+1))!}{(4(s+1)-2)!}\right)\;\cdot\\
&\qquad \cdot\;
\pFq{3}{1}{s+1+\tfrac{1}{4},s+1+\tfrac{2}{4},s+1+\tfrac{3}{4}}{\tfrac{3}{2}}{\lambda y^22^6}\,.
\end{aligned}
\end{equation}

Note in particular the appearance of Hermite polynomial expansion coefficients in the formulae for $\cH_{K,0}(\lambda;x,y)$ as presented in Theorem~\ref{thm:HPlacK} -- for example, the fraction of factorials appearing in the second to last line of~\eqref{eq:H40step2} coincides with the coefficient of the monomial $x^{4s+2}y$ of the Hermite polynomial $H_{4(s+1)}(x,y)$ (compare~\eqref{eq:HPtwo}),
\[
H_{4(s+1)}(x,y)=\sum_{k=0}^{2(s+1)}\frac{(4(s+1))!}{(4(s+1)-2k)!k!}\; x^{4(s+1)-2k}y^k\,.
\]

Combining the results of Theorem~\ref{thm:HPlacK} with the results of Proposition~\ref{prop:LshiftHLGF}, we have at our disposal explicit formulae for all higher order lacunary generating functions of the Hermite polynomials $H_n(x,y)$.

\begin{restatable}{cor}{HPlacKL}\label{cor:HPlacKL}
The $K$-tuple $L$-shifted lacunary generating functions $\cH_{K,L}(\lambda;x,y)$ read for $K=2T$ ($T\in \bZ_{\geq 1}$) and $L\in \bZ_{\geq 0}$
\begin{equation}
\begin{aligned}\label{eq:lacHPevenLshifted}
&\cH_{K=2T,L}(\lambda;x,y)\\
&\quad=
\sum_{s=0}^{\infty}\tfrac{\lambda^s}{s!}\;
\left(\sum_{q=0}^L q! \binom{L}{q}\binom{K s}{q}H_{L-q}(x,y)\;x^{K s-q}(2y)^q\right)
\;\cdot\\
&\quad\qquad \cdot\;
\pFq{(K-1)}{(T-1)}{%
\seq{s+\tfrac{j+1}{K}}{0\leq j\leq K-2}}{%
\seq{\tfrac{\ell+1}{T}}{0\leq \ell\leq T-2}}{\lambda(2Ky)^T}\\
&\quad 
+\sum_{\beta=1}^{T-1}\sum_{s=0}^{\infty}\tfrac{\lambda^{s+1}}{(s+1)!}\; 
\bigg(\sum_{q=0}^L q!\binom{L}{q}\binom{K (s+1)-2\beta}{q} H_{L-q}(x,y)\;\cdot\\
&\hphantom{+\sum_{\beta=1}^{T-1}\sum_{s\geq 0}\tfrac{\lambda^{s+1}}{(s+1)!}}\qquad\cdot\;x^{K (s+1)-2\beta-q}\; 2^q\;y^{\beta+q}\bigg)\;
\left(\tfrac{(K(s+1))!}{(K(s+1)-2\beta)!\beta!}\right)\;\cdot\\
&\hphantom{+\sum_{\beta=1}^{T-1}\sum_{s\geq 0}\tfrac{\lambda^{s+1}}{(s+1)!}}\quad\cdot\;
\pFq{(K-1)}{(T-1)}{\seq{s+1+\tfrac{j+1}{K}}{0\leq j\leq K-2}}{%
\seq{\tfrac{\beta+\ell+1}{T}}{\stackrel{0\leq\ell\leq T-1}{\ell\neq T-1-\beta}}}{\lambda(2Ky)^T}\,,
\end{aligned}
\end{equation}
and for $K=2T+1$ (with $T\in \bZ_{\geq1}$)
\begin{equation}
\begin{aligned}\label{eq:lacHPoddLshifted}
&\cH_{K=2T+1,L}(\lambda;x,y)\\
&\quad=
\sum_{s=0}^{\infty}\tfrac{\lambda^s}{s!}\;\left(\sum_{q=0}^L
q!\binom{L}{q}\binom{K s}{q}H_{L-q}(x,y)\;x^{K s-q}(2y)^q\right)\;\cdot\\
&\quad\qquad \cdot\;\pFq{(2K-2)}{(K-1)}{\seq{\tfrac{s}{2}+\tfrac{j+1}{2K}}{\stackrel{0\leq j\leq 2K-2}{j\neq K-1}}}{\seq{\tfrac{\ell+1}{K}}{0\leq \ell \leq K-2}}{\tfrac{\lambda^2(4Ky)^K}{4}}\\
&\quad 
+\sum_{\beta=1}^{T}\sum_{s=0}^{\infty}\tfrac{\lambda^{s+1}}{(s+1)!}\; 
\bigg(\sum_{q=0}^L q!\binom{L}{q}\binom{K(s+1)-2\beta}{q}H_{L-q}(x,y)\; \cdot\\
&\hphantom{+\sum_{\beta=1}^{T}\sum_{s\geq 0}\tfrac{\lambda^{s+1}}{(s+1)!}\; 
}\cdot\;x^{K(s+1)-2\beta-q}\; 2^q\;y^{\beta+q}\bigg)\;\left(\tfrac{(K(s+1))!}{(K(s+1)-2\beta)!\beta!}\right)\;\cdot\\
&\hphantom{+\sum_{\beta=1}^{T}\sum_{s\geq 0}\tfrac{\lambda^{s+1}}{(s+1)!}\; 
}\cdot\;\pFq{(2K-2)}{(K-1)}{\seq{\tfrac{s+1}{2}+\tfrac{j+1}{2K}}{\stackrel{0\leq j\leq 2K-2}{j\neq K-1}}}{\seq{\tfrac{\beta+\ell+1}{K}}{\stackrel{0\leq \ell\leq K-1}{\ell\neq K-1-\beta}}}{\tfrac{\lambda^2(4Ky)^K}{4}}\\
&\quad 
+\sum_{\beta=1}^{T}
\sum_{s=0}^{\infty}\tfrac{\lambda^{s+2}}{(s+2)!}\; 
\bigg(\sum_{q=0}^L q!\binom{L}{q}\binom{K(s+2)-2(T+\beta)}{q} H_{L-q}(x,y) \;
\cdot\\
&\hphantom{+\sum_{\beta=1}^{T}
\sum_{s\geq 0}\tfrac{\lambda^{s+2}}{(s+2)!}\; }\quad\cdot
\;x^{K(s+2)-2(T+\beta)-q}\;2^q\; y^{T+\beta+q}\bigg)\;
\left(\tfrac{(K(s+2))!}{(K(s+2)-2(T+\beta))!(T+\beta)!}\right)\;\cdot\\
&\hphantom{+\sum_{\beta=1}^{T}
\sum_{s\geq 0}\tfrac{\lambda^{s+2}}{(s+2)!}\;}\quad\cdot\;
\pFq{(2K-2)}{(K-1)}{\seq{\tfrac{s+2}{2}+\tfrac{j+1}{2K}}{\stackrel{0\leq j\leq 2K-2}{j\neq K-1}}}{
\seq{\tfrac{T+\beta+\ell+1}{K}}{\stackrel{0\leq \ell\leq K-1}{\ell\neq T-\beta}}}{\tfrac{\lambda^2(4Ky)^K}{4}}\,.
\end{aligned}
\end{equation}
\end{restatable}
\begin{proof}
See Appendix~\ref{app:proofLacKL} for the details. 
\end{proof}

For the readers' convenience, we present some explicit results for $\cH_{K,0}(\lambda;x,y)$ in Table~\ref{tab:HPlacFull}, as well as several examples of shifted lacunary generating functions $\cH_{K,L}(\lambda;x,y)$ (for $K=3,4$ and $L=1,2,3$) obtained via Corollary~\ref{cor:HPlacKL} in Table~\ref{tab:HPlacShift}. To the best of our knowledge, all of these results past $K=2$ are new.

\section{Conclusion}

Taking advantage of the summability properties of exponential generating functions of polynomial sequences, we provide novel re-summation techniques combined with operational methods in order to compute higher order lacunary generating functions of polynomials of all types. The application of these techniques to the case of two-variable Hermite polynomials $H_n(x,y)$ furnished explicit formulae for an infinite set of $K$-tuple $L$-shifted lacunary generating functions $\cH_{K,L}(\lambda;x,y)$, which subsume the known expressions for $\cH_{1,L}(\lambda;x,y)$~\cite{prudnikov1992integrals}, $\cH_{2,0}(\lambda;x,y)$~\cite{foata1981some,foataStrehl1984,dattoli2017operational}
 and $\cH_{3,0}(\lambda;x,y)$~\cite{gessel2005triple}. For the convenience of readers who would like to perform experimental mathematics around our methods, we provide some explicit examples for $\cH_{K,0}(\lambda;x,y)$ for $K=1,\dotsc,10$ in Table~\ref{tab:HPlacFull} as well as some Maple code in Appendix~\ref{app:Maple}.

{\footnotesize

}

\footnotesize

\begin{landscape}
\begin{longtable}{L}
  \caption{\small\ Lacunary generating functions $\mathcal{H}_{K,0}(\lambda;x;y)$ for the bi-variate Hermite polynomials $H_n(x,y)$ and $K=1\dotsc 10$.\label{tab:HPlacFull}}\\
\hline\\
\endfirsthead
{\tablename\ \thetable\  \textit{-- continued from previous page}} \\
\hline\\
\endhead
\hline \textit{Continued on next page} \\
\endfoot
\hline
\endlastfoot
\mathcal{H}_{1,0}(\lambda;x,y)=e^{\lambda x+ \lambda^2 y}\\[1.5em]\\
\mathcal{H}_{2,0}(\lambda;x,y)=\sum_{s=0}^{\infty} \tfrac{\lambda^s}{s!} \; x^{2s} \; \pFq{1}{0}{s+\frac{1}{2}}{-}{\lambda  y\; 2^2}
=\frac{1}{\sqrt{1-4\lambda y}}e^{\frac{x^2\lambda}{\sqrt{1-4\lambda y}}}\\[1.5em]
\\
\mathcal{H}_{3,0}(\lambda;x,y)=\sum_{s=0}^{\infty} \tfrac{\lambda^s}{s!} \; x^{3s} \; \pFq{4}{2}{\frac{s}{2} + \frac{1}{6},\frac{s}{2} + \frac{1}{3},\frac{s}{2} + \frac{2}{3},\frac{s}{2} + \frac{5}{6}}{\frac{1}{3},\frac{2}{3}}{ \lambda^2 y^{3}\; 2^4 3^3}\\[1.5em]
\hphantom{\mathcal{H}_{3,0}(\lambda;x,y)=} + \sum_{s=0}^{\infty} \tfrac{\lambda^{s+1}}{(s+1)!} \; x^{3 s+ 1} \; y \; \left(\tfrac{(3(s+1))!}{(3(s+1)- 2)! }\right) \; \pFq{4}{2}{\frac{s+1}{2} + \frac{1}{6},\frac{s+1}{2} + \frac{1}{3},\frac{s+1}{2} + \frac{2}{3},\frac{s+1}{2} + \frac{5}{6}}{\frac{2}{3},\frac{4}{3}}{ \lambda^2 y^{3} \; 2^4 3^3}\\[1.5em]
\hphantom{\mathcal{H}_{3,0}(\lambda;x,y)=}+\sum_{s=0}^{\infty} \tfrac{\lambda^{s+2}}{(s+2)!}\; x^{3s + 2} \; y^2 \; \left(\tfrac{(3(s+2))!}{(3(s+2) - 4)!2!}\right) \; \pFq{4}{2}{\frac{s+2}{2} + \frac{1}{6},\frac{s+2}{2} + \frac{1}{3},\frac{s+2}{2} + \frac{2}{3},\frac{s+2}{2} + \frac{5}{6}}{\frac{4}{3},\frac{5}{3}}{\lambda^2 y^{3} \; 2^4 3^3}\\[1.5em]
\\
\mathcal{H}_{4,0}(\lambda;x,y)=\sum_{s=0}^{\infty} \tfrac{\lambda^s}{s!} \; x^{4s} \; \pFq{3}{1}{s+\frac{1}{4},s+\frac{1}{2},s+\frac{3}{4}}{\frac{1}{2}}{\lambda  y^2\; 2^6}\\[1.5em]
\hphantom{\mathcal{H}_{4,0}(\lambda;x,y)=}+\sum_{s=0}^{\infty} \tfrac{\lambda^{s+1}}{(s+1)!} \;x^{4 s+2}\;y\; \left(\tfrac{(4(s+1))!}{(4(s+1) -2)!}\right)\; \pFq{3}{1}{(s+1) + \frac{1}{4},(s+1) + \frac{1}{2},(s+1) + \frac{3}{4}}{\frac{3}{2}}{\lambda  y^2\; 2^6}\\[1.5em]
\\
\mathcal{H}_{5,0}(\lambda;x,y)=\sum_{s=0}^{\infty} \tfrac{\lambda^s}{s!} \; x^{5s} \; \pFq{8}{4}{\frac{s}{2} + \frac{1}{10},\frac{s}{2} + \frac{1}{5},\frac{s}{2} + \frac{3}{10},\frac{s}{2} + \frac{2}{5},\frac{s}{2} + \frac{3}{5},\frac{s}{2} + \frac{7}{10},\frac{s}{2} + \frac{4}{5},\frac{s}{2} + \frac{9}{10}}{\frac{1}{5},\frac{2}{5},\frac{3}{5},\frac{4}{5}}{ \lambda^2 y^{5}\; 2^8 5^5}\\[1.5em]
\hphantom{\mathcal{H}_{5,0}(\lambda;x,y)=} + \sum_{s=0}^{\infty} \tfrac{\lambda^{s+1}}{(s+1)!} \; x^{5 s+ 3} \; y \; \left(\tfrac{(5(s+1))!}{(5(s+1)- 2)! }\right) \; \pFq{8}{4}{\frac{s+1}{2} + \frac{1}{10},\frac{s+1}{2} + \frac{1}{5},\frac{s+1}{2} + \frac{3}{10},\frac{s+1}{2} + \frac{2}{5},\frac{s+1}{2} + \frac{3}{5},\frac{s+1}{2} + \frac{7}{10},\frac{s+1}{2} + \frac{4}{5},\frac{s+1}{2} + \frac{9}{10}}{\frac{2}{5},\frac{3}{5},\frac{4}{5},\frac{6}{5}}{ \lambda^2 y^{5} \; 2^8 5^5}\\[1.5em]
\hphantom{\mathcal{H}_{5,0}(\lambda;x,y)=} + \sum_{s=0}^{\infty} \tfrac{\lambda^{s+1}}{(s+1)!} \; x^{5 s+ 1} \; y^2 \; \left(\tfrac{(5(s+1))!}{(5(s+1)- 4)! 2!}\right) \; \pFq{8}{4}{\frac{s+1}{2} + \frac{1}{10},\frac{s+1}{2} + \frac{1}{5},\frac{s+1}{2} + \frac{3}{10},\frac{s+1}{2} + \frac{2}{5},\frac{s+1}{2} + \frac{3}{5},\frac{s+1}{2} + \frac{7}{10},\frac{s+1}{2} + \frac{4}{5},\frac{s+1}{2} + \frac{9}{10}}{\frac{3}{5},\frac{4}{5},\frac{6}{5},\frac{7}{5}}{ \lambda^2 y^{5} \; 2^8 5^5}\\[1.5em]
\hphantom{\mathcal{H}_{5,0}(\lambda;x,y)=}+\sum_{s=0}^{\infty} \tfrac{\lambda^{s+2}}{(s+2)!}\; x^{5s + 4} \; y^3 \; \left(\tfrac{(5(s+2))!}{(5(s+2) - 6)!3!}\right) \; \pFq{8}{4}{\frac{s+2}{2} + \frac{1}{10},\frac{s+2}{2} + \frac{1}{5},\frac{s+2}{2} + \frac{3}{10},\frac{s+2}{2} + \frac{2}{5},\frac{s+2}{2} + \frac{3}{5},\frac{s+2}{2} + \frac{7}{10},\frac{s+2}{2} + \frac{4}{5},\frac{s+2}{2} + \frac{9}{10}}{\frac{4}{5},\frac{6}{5},\frac{7}{5},\frac{8}{5}}{\lambda^2 y^{5} \; 2^8 5^5}\\[1.5em]
\hphantom{\mathcal{H}_{5,0}(\lambda;x,y)=}+\sum_{s=0}^{\infty} \tfrac{\lambda^{s+2}}{(s+2)!}\; x^{5s + 2} \; y^4 \; \left(\tfrac{(5(s+2))!}{(5(s+2) - 8)!4!}\right) \; \pFq{8}{4}{\frac{s+2}{2} + \frac{1}{10},\frac{s+2}{2} + \frac{1}{5},\frac{s+2}{2} + \frac{3}{10},\frac{s+2}{2} + \frac{2}{5},\frac{s+2}{2} + \frac{3}{5},\frac{s+2}{2} + \frac{7}{10},\frac{s+2}{2} + \frac{4}{5},\frac{s+2}{2} + \frac{9}{10}}{\frac{6}{5},\frac{7}{5},\frac{8}{5},\frac{9}{5}}{\lambda^2 y^{5} \; 2^8 5^5}\\[1.5em]
\\
\mathcal{H}_{6,0}(\lambda;x,y)=\sum_{s=0}^{\infty} \tfrac{\lambda^s}{s!} \; x^{6s} \; \pFq{5}{2}{s+\frac{1}{6},s+\frac{1}{3},s+\frac{1}{2},s+\frac{2}{3},s+\frac{5}{6}}{\frac{1}{3},\frac{2}{3}}{\lambda  y^3\; 2^6 3^3}\\[1.5em]
\hphantom{\mathcal{H}_{6,0}(\lambda;x,y)=}+\sum_{s=0}^{\infty} \tfrac{\lambda^{s+1}}{(s+1)!} \;x^{6 s+4}\;y\; \left(\tfrac{(6(s+1))!}{(6(s+1) -2)!}\right)\; \pFq{5}{2}{(s+1) + \frac{1}{6},(s+1) + \frac{1}{3},(s+1) + \frac{1}{2},(s+1) + \frac{2}{3},(s+1) + \frac{5}{6}}{\frac{2}{3},\frac{4}{3}}{\lambda  y^3\; 2^6 3^3}\\[1.5em]
\hphantom{\mathcal{H}_{6,0}(\lambda;x,y)=}+\sum_{s=0}^{\infty} \tfrac{\lambda^{s+1}}{(s+1)!} \;x^{6 s+2}\;y^2\; \left(\tfrac{(6(s+1))!}{2!(6(s+1) -4)!}\right)\; \pFq{5}{2}{(s+1) + \frac{1}{6},(s+1) + \frac{1}{3},(s+1) + \frac{1}{2},(s+1) + \frac{2}{3},(s+1) + \frac{5}{6}}{\frac{4}{3},\frac{5}{3}}{\lambda  y^3\; 2^6 3^3}\\[1.5em]
\\
\mathcal{H}_{7,0}(\lambda;x,y)
=\sum_{s=0}^{\infty} \tfrac{\lambda^s}{s!} \; x^{7s} \; 
\pFq{12}{6}{\frac{s}{2} + \frac{1}{14},\frac{s}{2} + \frac{2}{14},
\dotsc\,,
\frac{s}{2} + \frac{13}{14}}{\frac{1}{7},\frac{2}{7},\frac{3}{7},\frac{4}{7},\frac{5}{7},\frac{6}{7}}{ \lambda^2 y^{7}\; 2^{12} 7^7}\\[1.5em]
\hphantom{\mathcal{H}_{7,0}(\lambda;x,y)=} 
+ \sum_{s=0}^{\infty} \tfrac{\lambda^{s+1}}{(s+1)!} \; x^{7 s+ 5} \; y \; \left(\tfrac{(7(s+1))!}{(7(s+1)- 2)! }\right) \; 
\pFq{12}{6}{\frac{s+1}{2} + \frac{1}{14},\frac{s+1}{2} + \frac{2}{14},
\dotsc\,,
\frac{s+1}{2} + \frac{13}{14}}{\frac{2}{7},\frac{3}{7},\frac{4}{7},\frac{5}{7},\frac{6}{7},\frac{8}{7}}{ \lambda^2 y^{7} \; 2^{12} 7^7}\\[1.5em]
\hphantom{\mathcal{H}_{7,0}(\lambda;x,y)=} + \sum_{s=0}^{\infty} \tfrac{\lambda^{s+1}}{(s+1)!} \; x^{7 s+ 3} \; y^2 \; \left(\tfrac{(7(s+1))!}{(7(s+1)- 4)! 2!}\right) \; 
\pFq{12}{6}{\frac{s+1}{2} + \frac{1}{14},\frac{s+1}{2} + \frac{2}{14},
\dotsc\,,
\frac{s+1}{2} + \frac{13}{14}}{\frac{3}{7},\frac{4}{7},\frac{5}{7},\frac{6}{7},\frac{8}{7},\frac{9}{7}}{ \lambda^2 y^{7} \; 2^{12} 7^7}\\[1.5em]
\hphantom{\mathcal{H}_{7,0}(\lambda;x,y)=} + \sum_{s=0}^{\infty} \tfrac{\lambda^{s+1}}{(s+1)!} \; x^{7 s+ 1} \; y^3 \; \left(\tfrac{(7(s+1))!}{(7(s+1)- 6)! 3!}\right) \; 
\pFq{12}{6}{\frac{s+1}{2} + \frac{1}{14},\frac{s+1}{2} + \frac{2}{14},
\dotsc\,,
\frac{s+1}{2} + \frac{13}{14}}{\frac{4}{7},\frac{5}{7},\frac{6}{7},\frac{8}{7},\frac{9}{7},\frac{10}{7}}{ \lambda^2 y^{7} \; 2^{12} 7^7}\\[1.5em]
\hphantom{\mathcal{H}_{7,0}(\lambda;x,y)=}+\sum_{s=0}^{\infty} \tfrac{\lambda^{s+2}}{(s+2)!}\; x^{7s + 6} \; y^4 \; \left(\tfrac{(7(s+2))!}{(7(s+2) - 8)!4!}\right) \; 
\pFq{12}{6}{\frac{s+2}{2} + \frac{1}{14},\frac{s+2}{2} + \frac{2}{14},
\dotsc\,,
\frac{s+2}{2} + \frac{13}{14}}{\frac{5}{7},\frac{6}{7},\frac{8}{7},\frac{9}{7},\frac{10}{7},\frac{11}{7}}{\lambda^2 y^{7} \; 2^{12} 7^7}\\[1.5em]
\hphantom{\mathcal{H}_{7,0}(\lambda;x,y)=}+\sum_{s=0}^{\infty} \tfrac{\lambda^{s+2}}{(s+2)!}\; x^{7s + 4} \; y^5 \; \left(\tfrac{(7(s+2))!}{(7(s+2) - 10)!5!}\right) \;
\pFq{12}{6}{\frac{s+2}{2} + \frac{1}{14},\frac{s+2}{2} + \frac{2}{14},
\dotsc\,,
\frac{s+2}{2} + \frac{13}{14}}{\frac{6}{7},\frac{8}{7},\frac{9}{7},\frac{10}{7},\frac{11}{7},\frac{12}{7}}{\lambda^2 y^{7} \; 2^{12} 7^7}\\[1.5em]
\hphantom{\mathcal{H}_{7,0}(\lambda;x,y)=}+\sum_{s=0}^{\infty} \tfrac{\lambda^{s+2}}{(s+2)!}\; x^{7s + 2} \; y^6 \; \left(\tfrac{(7(s+2))!}{(7(s+2) - 12)!6!}\right) \; 
\pFq{12}{6}{\frac{s+2}{2} + \frac{1}{14},\frac{s+2}{2} + \frac{2}{14},
\dotsc\,,
\frac{s+2}{2} + \frac{13}{14}}{\frac{8}{7},\frac{9}{7},\frac{10}{7},\frac{11}{7},\frac{12}{7},\frac{13}{7}}{\lambda^2 y^{7} \; 2^{12} 7^7}\\[1.5em]
\\
\mathcal{H}_{8,0}(\lambda;x,y)=\sum_{s=0}^{\infty} \tfrac{\lambda^s}{s!} \; x^{8s} \; \pFq{7}{3}{s+\frac{1}{8},s+\frac{1}{4},s+\frac{3}{8},s+\frac{1}{2},s+\frac{5}{8},s+\frac{3}{4},s+\frac{7}{8}}{\frac{1}{4},\frac{1}{2},\frac{3}{4}}{\lambda  y^4\; 2^{16}}\\[1.5em]
\hphantom{\mathcal{H}_{8,0}(\lambda;x,y)=}+\sum_{s=0}^{\infty} \tfrac{\lambda^{s+1}}{(s+1)!} \;x^{8 s+6}\;y\; \left(\tfrac{(8(s+1))!}{(8(s+1) -2)!}\right)\; \pFq{7}{3}{(s+1) + \frac{1}{8},(s+1) + \frac{1}{4},(s+1) + \frac{3}{8},(s+1) + \frac{1}{2},(s+1) + \frac{5}{8},(s+1) + \frac{3}{4},(s+1) + \frac{7}{8}}{\frac{1}{2},\frac{3}{4},\frac{5}{4}}{\lambda  y^4\; 2^{16}}\\[1.5em]
\hphantom{\mathcal{H}_{8,0}(\lambda;x,y)=}+\sum_{s=0}^{\infty} \tfrac{\lambda^{s+1}}{(s+1)!} \;x^{8 s+4}\;y^2\; \left(\tfrac{(8(s+1))!}{2!(8(s+1) -4)!}\right)\; \pFq{7}{3}{(s+1) + \frac{1}{8},(s+1) + \frac{1}{4},(s+1) + \frac{3}{8},(s+1) + \frac{1}{2},(s+1) + \frac{5}{8},(s+1) + \frac{3}{4},(s+1) + \frac{7}{8}}{\frac{3}{4},\frac{5}{4},\frac{3}{2}}{\lambda  y^4\; 2^{16}}\\[1.5em]
\hphantom{\mathcal{H}_{8,0}(\lambda;x,y)=}+\sum_{s=0}^{\infty} \tfrac{\lambda^{s+1}}{(s+1)!} \;x^{8 s+2}\;y^3\; \left(\tfrac{(8(s+1))!}{3!(8(s+1) -6)!}\right)\; \pFq{7}{3}{(s+1) + \frac{1}{8},(s+1) + \frac{1}{4},(s+1) + \frac{3}{8},(s+1) + \frac{1}{2},(s+1) + \frac{5}{8},(s+1) + \frac{3}{4},(s+1) + \frac{7}{8}}{\frac{5}{4},\frac{3}{2},\frac{7}{4}}{\lambda  y^4\; 2^{16}}\\[1.5em]
\\
\mathcal{H}_{9,0}(\lambda;x,y)=\sum_{s=0}^{\infty} \tfrac{\lambda^s}{s!} \; x^{9s} \; 
\pFq{16}{8}{\frac{s}{2} + \frac{1}{18},\frac{s}{2} + \frac{2}{18},
\dotsc\,,
\frac{s}{2} + \frac{17}{18}}{\frac{1}{9},\frac{2}{9},\frac{1}{3},\frac{4}{9},\frac{5}{9},\frac{2}{3},\frac{7}{9},\frac{8}{9}}{ \lambda^2 y^{9}\; 2^{16} 3^{18}}\\[1.5em]
\hphantom{\mathcal{H}_{9,0}(\lambda;x,y)=} + \sum_{s=0}^{\infty} \tfrac{\lambda^{s+1}}{(s+1)!} \; x^{9 s+ 7} \; y \; \left(\tfrac{(9(s+1))!}{(9(s+1)- 2)! }\right) \; 
\pFq{16}{8}{\frac{s+1}{2} + \frac{1}{18},\frac{s+1}{2} + \frac{2}{18},
\dotsc\,,
\frac{s+1}{2} + \frac{17}{18}}{\frac{2}{9},\frac{1}{3},\frac{4}{9},\frac{5}{9},\frac{2}{3},\frac{7}{9},\frac{8}{9},\frac{10}{9}}{ \lambda^2 y^{9} \; 2^{16} 3^{18}}\\[1.5em]
\hphantom{\mathcal{H}_{9,0}(\lambda;x,y)=} + \sum_{s=0}^{\infty} \tfrac{\lambda^{s+1}}{(s+1)!} \; x^{9 s+ 5} \; y^2 \; \left(\tfrac{(9(s+1))!}{(9(s+1)- 4)! 2!}\right) \; 
\pFq{16}{8}{\frac{s+1}{2} + \frac{1}{18},\frac{s+1}{2} + \frac{2}{18},
\dotsc\,,
\frac{s+1}{2} + \frac{17}{18}}{\frac{1}{3},\frac{4}{9},\frac{5}{9},\frac{2}{3},\frac{7}{9},\frac{8}{9},\frac{10}{9},\frac{11}{9}}{ \lambda^2 y^{9} \; 2^{16} 3^{18}}\\[1.5em]
\hphantom{\mathcal{H}_{9,0}(\lambda;x,y)=} + \sum_{s=0}^{\infty} \tfrac{\lambda^{s+1}}{(s+1)!} \; x^{9 s+ 3} \; y^3 \; \left(\tfrac{(9(s+1))!}{(9(s+1)- 6)! 3!}\right) \; 
\pFq{16}{8}{\frac{s+1}{2} + \frac{1}{18},\frac{s+1}{2} + \frac{2}{18},
\dotsc\,,
\frac{s+1}{2} + \frac{17}{18}}{\frac{4}{9},\frac{5}{9},\frac{2}{3},\frac{7}{9},\frac{8}{9},\frac{10}{9},\frac{11}{9},\frac{4}{3}}{ \lambda^2 y^{9} \; 2^{16} 3^{18}}\\[1.5em]
\hphantom{\mathcal{H}_{9,0}(\lambda;x,y)=} + \sum_{s=0}^{\infty} \tfrac{\lambda^{s+1}}{(s+1)!} \; x^{9 s+ 1} \; y^4 \; \left(\tfrac{(9(s+1))!}{(9(s+1)- 8)! 4!}\right) \; 
\pFq{16}{8}{\frac{s+1}{2} + \frac{1}{18},\frac{s+1}{2} + \frac{2}{18},
\dotsc\,,
\frac{s+1}{2} + \frac{17}{18}}{\frac{5}{9},\frac{2}{3},\frac{7}{9},\frac{8}{9},\frac{10}{9},\frac{11}{9},\frac{4}{3},\frac{13}{9}}{ \lambda^2 y^{9} \; 2^{16} 3^{18}}\\[1.5em]
\hphantom{\mathcal{H}_{9,0}(\lambda;x,y)=}+\sum_{s=0}^{\infty} \tfrac{\lambda^{s+2}}{(s+2)!}\; x^{9s + 8} \; y^5 \; \left(\tfrac{(9(s+2))!}{(9(s+2) - 10)!5!}\right) \; 
\pFq{16}{8}{\frac{s+2}{2} + \frac{1}{18},\frac{s+2}{2} + \frac{2}{18},
\dotsc\,,
\frac{s+2}{2} + \frac{17}{18}}{\frac{2}{3},\frac{7}{9},\frac{8}{9},\frac{10}{9},\frac{11}{9},\frac{4}{3},\frac{13}{9},\frac{14}{9}}{\lambda^2 y^{9} \; 2^{16} 3^{18}}\\[1.5em]
\hphantom{\mathcal{H}_{9,0}(\lambda;x,y)=}+\sum_{s=0}^{\infty} \tfrac{\lambda^{s+2}}{(s+2)!}\; x^{9s + 6} \; y^6 \; \left(\tfrac{(9(s+2))!}{(9(s+2) - 12)!6!}\right) \; 
\pFq{16}{8}{\frac{s+2}{2} + \frac{1}{18},\frac{s+2}{2} + \frac{2}{18},
\dotsc\,,
\frac{s+2}{2} + \frac{17}{18}}{\frac{7}{9},\frac{8}{9},\frac{10}{9},\frac{11}{9},\frac{4}{3},\frac{13}{9},\frac{14}{9},\frac{5}{3}}{\lambda^2 y^{9} \; 2^{16} 3^{18}}\\[1.5em]
\hphantom{\mathcal{H}_{9,0}(\lambda;x,y)=}+\sum_{s=0}^{\infty} \tfrac{\lambda^{s+2}}{(s+2)!}\; x^{9s + 4} \; y^7 \; \left(\tfrac{(9(s+2))!}{(9(s+2) - 14)!7!}\right) \; 
\pFq{16}{8}{\frac{s+2}{2} + \frac{1}{18},\frac{s+2}{2} + \frac{2}{18},
\dotsc\,,
\frac{s+2}{2} + \frac{17}{18}}{\frac{8}{9},\frac{10}{9},\frac{11}{9},\frac{4}{3},\frac{13}{9},\frac{14}{9},\frac{5}{3},\frac{16}{9}}{\lambda^2 y^{9} \; 2^{16} 3^{18}}\\[1.5em]
\hphantom{\mathcal{H}_{9,0}(\lambda;x,y)=}+\sum_{s=0}^{\infty} \tfrac{\lambda^{s+2}}{(s+2)!}\; x^{9s + 2} \; y^8 \; \left(\tfrac{(9(s+2))!}{(9(s+2) - 16)!8!}\right) \; 
\pFq{16}{8}{\frac{s+2}{2} + \frac{1}{18},\frac{s+2}{2} + \frac{2}{18},
\dotsc\,,
\frac{s+2}{2} + \frac{17}{18}}{\frac{10}{9},\frac{11}{9},\frac{4}{3},\frac{13}{9},\frac{14}{9},\frac{5}{3},\frac{16}{9},\frac{17}{9}}{\lambda^2 y^{9} \; 2^{16} 3^{18}}\\[1.5em]
\pagebreak
\mathcal{H}_{10,0}(\lambda;x,y)=\sum_{s=0}^{\infty} \tfrac{\lambda^s}{s!} \; x^{10s} \; 
\pFq{9}{4}{s+\frac{1}{10},s+\frac{2}{10},\dotsc\,,s+\frac{9}{10}}{\frac{1}{5},\frac{2}{5},\frac{3}{5},\frac{4}{5}}{\lambda  y^5\; 2^{10} 5^5}\\[1.5em]
\hphantom{\mathcal{H}_{10,0}(\lambda;x,y)=}+\sum_{s=0}^{\infty} \tfrac{\lambda^{s+1}}{(s+1)!} \;x^{10 s+8}\;y\; \left(\tfrac{(10(s+1))!}{(10(s+1) -2)!}\right)\; 
\pFq{9}{4}{(s+1) + \frac{1}{10},(s+1) + \frac{2}{10},\dotsc\,,(s+1) + \frac{9}{10}}{\frac{2}{5},\frac{3}{5},\frac{4}{5},\frac{6}{5}}{\lambda  y^5\; 2^{10} 5^5}\\[1.5em]
\hphantom{\mathcal{H}_{10,0}(\lambda;x,y)=}+\sum_{s=0}^{\infty} \tfrac{\lambda^{s+1}}{(s+1)!} \;x^{10 s+6}\;y^2\; \left(\tfrac{(10(s+1))!}{2!(10(s+1) -4)!}\right)\; 
\pFq{9}{4}{(s+1) + \frac{1}{10},(s+1) + \frac{2}{10},\dotsc\,,(s+1) + \frac{9}{10}}{\frac{3}{5},\frac{4}{5},\frac{6}{5},\frac{7}{5}}{\lambda  y^5\; 2^{10} 5^5}\\[1.5em]
\hphantom{\mathcal{H}_{10,0}(\lambda;x,y)=}+\sum_{s=0}^{\infty} \tfrac{\lambda^{s+1}}{(s+1)!} \;x^{10 s+4}\;y^3\; \left(\tfrac{(10(s+1))!}{3!(10(s+1) -6)!}\right)\; 
\pFq{9}{4}{(s+1) + \frac{1}{10},(s+1) + \frac{2}{10},\dotsc\,,(s+1) + \frac{9}{10}}{\frac{4}{5},\frac{6}{5},\frac{7}{5},\frac{8}{5}}{\lambda  y^5\; 2^{10} 5^5}\\[1.5em]
\hphantom{\mathcal{H}_{10,0}(\lambda;x,y)=}+\sum_{s=0}^{\infty} \tfrac{\lambda^{s+1}}{(s+1)!} \;x^{10 s+2}\;y^4\; \left(\tfrac{(10(s+1))!}{4!(10(s+1) -8)!}\right)\; 
\pFq{9}{4}{(s+1) + \frac{1}{10},(s+1) + \frac{2}{10},\dotsc\,,(s+1) + \frac{9}{10}}{\frac{6}{5},\frac{7}{5},\frac{8}{5},\frac{9}{5}}{\lambda  y^5\; 2^{10} 5^5}\\[1.5em]
\end{longtable}
\end{landscape}

\begin{landscape}
\begin{longtable}{L}
  \caption{\small\ Shifted exponential lacunary generating functions $\mathcal{H}_{K,L}(\lambda;x;y)$ for the bi-variate Hermite polynomials $H_n(x,y)$, $K=3,4$ and $L=1,2,3$.\label{tab:HPlacShift}}\\
\hline\\
\endfirsthead
{\tablename\ \thetable\  \textit{-- continued from previous page}} \\
\hline\\
\endhead
\hline \textit{Continued on next page} \\
\endfoot
\hline
\endlastfoot
\mathcal{H}_{3,1}(\lambda;x,y)=\sum_{s=0}^{\infty} \tfrac{\lambda^s}{s!} \; 
\left(
x^{3 s+1}+  x^{3 s-1} \binom{3 s}{1}2 y
\right) \; \pFq{4}{2}{\frac{s}{2} + \frac{1}{6},\frac{s}{2} + \frac{1}{3},\frac{s}{2} + \frac{2}{3},\frac{s}{2} + \frac{5}{6}}{\frac{1}{3},\frac{2}{3}}{ \lambda^2 y^{3}\; 2^4 3^3}\\[1.5em]
\hphantom{\mathcal{H}_{3,0}(\lambda;x,y)=} + \sum_{s=0}^{\infty} \tfrac{\lambda^{s+1}}{(s+1)!} \; 
\left(
x^{3 (s+1)-1}y
+ x^{3 (s+1)-3} \binom{3 (s+1)-2}{1}2 y^2 
\right)\; \left(\tfrac{(3(s+1))!}{(3(s+1)- 2)! }\right) \; \pFq{4}{2}{\frac{s+1}{2} + \frac{1}{6},\frac{s+1}{2} + \frac{1}{3},\frac{s+1}{2} + \frac{2}{3},\frac{s+1}{2} + \frac{5}{6}}{\frac{2}{3},\frac{4}{3}}{ \lambda^2 y^{3} \; 2^4 3^3}\\[1.5em]
\hphantom{\mathcal{H}_{3,0}(\lambda;x,y)=}+\sum_{s=0}^{\infty} \tfrac{\lambda^{s+2}}{(s+2)!}\; 
\left(
x^{3 (s+2)-3}y^2
+x^{3 (s+2)-5} \binom{3 (s+2)-4}{1}2 y^3
\right) \; \left(\tfrac{(3(s+2))!}{(3(s+2) - 4)!2!}\right) \; \pFq{4}{2}{\frac{s+2}{2} + \frac{1}{6},\frac{s+2}{2} + \frac{1}{3},\frac{s+2}{2} + \frac{2}{3},\frac{s+2}{2} + \frac{5}{6}}{\frac{4}{3},\frac{5}{3}}{\lambda^2 y^{3} \; 2^4 3^3}\\[1.5em]
\mathcal{H}_{3,2}(\lambda;x,y)
=\sum_{s=0}^{\infty} \tfrac{\lambda^s}{s!} \; 
\left(
x^{3 s} \left(x^2+2 y\right)
+ x^{3 s} \binom{3 s}{1}4 y
+ x^{3 s-2} \binom{3 s}{2} 2!\cdot 4 y^2
\right) \; \pFq{4}{2}{\frac{s}{2} + \frac{1}{6},\frac{s}{2} + \frac{1}{3},\frac{s}{2} + \frac{2}{3},\frac{s}{2} + \frac{5}{6}}{\frac{1}{3},\frac{2}{3}}{ \lambda^2 y^{3}\; 2^4 3^3}\\[1.5em]
\hphantom{\mathcal{H}_{3,0}(\lambda;x,y)=} + \sum_{s=0}^{\infty} \tfrac{\lambda^{s+1}}{(s+1)!} \; 
\left( 
 x^{3 (s+1)-2}y \left(x^2+2 y\right)
+ x^{3 (s+1)-2} \binom{3 (s+1)-2}{1}4 y^2
+ x^{3 (s+1)-4} \binom{3 (s+1)-2}{2}2!\cdot 4 y^3 \right)\; \left(\tfrac{(3(s+1))!}{(3(s+1)- 2)! }\right) \; \cdot\\
\hphantom{\mathcal{H}_{3,0}(\lambda;x,y)=} \quad \cdot\;\pFq{4}{2}{\frac{s+1}{2} + \frac{1}{6},\frac{s+1}{2} + \frac{1}{3},\frac{s+1}{2} + \frac{2}{3},\frac{s+1}{2} + \frac{5}{6}}{\frac{2}{3},\frac{4}{3}}{ \lambda^2 y^{3} \; 2^4 3^3}\\[1.5em]
\hphantom{\mathcal{H}_{3,0}(\lambda;x,y)=}
+\sum_{s=0}^{\infty} \tfrac{\lambda^{s+2}}{(s+2)!}\; 
\left( 
x^{3 (s+2)-4}y^2 \left(x^2+2 y\right)
+ x^{3 (s+2)-4} \binom{3 (s+2)-4}{1}4 y^3
+ x^{3 (s+2)-6} \binom{3 (s+2)-4}{2} 2!\cdot 4 y^4\right) \; \left(\tfrac{(3(s+2))!}{(3(s+2) - 4)!2!}\right) \; \cdot\\
\hphantom{\mathcal{H}_{3,0}(\lambda;x,y)=}\quad \cdot \pFq{4}{2}{\frac{s+2}{2} + \frac{1}{6},\frac{s+2}{2} + \frac{1}{3},\frac{s+2}{2} + \frac{2}{3},\frac{s+2}{2} + \frac{5}{6}}{\frac{4}{3},\frac{5}{3}}{\lambda^2 y^{3} \; 2^4 3^3}\\[1.5em]
\mathcal{H}_{3,3}(\lambda;x,y)=\sum_{s=0}^{\infty} \tfrac{\lambda^s}{s!} \; \left(
x^{3 s} \left(x^3+6 x y\right)
+ x^{3 s-1}\binom{3 s}{1}6 y \left(x^2+2 y\right) 
+ x^{3 s-1} \binom{3 s}{2}2!\cdot 12 y^2
+ x^{3 s-3} \binom{3 s}{3}3!\cdot 8 y^3
\right) \; \pFq{4}{2}{\frac{s}{2} + \frac{1}{6},\frac{s}{2} + \frac{1}{3},\frac{s}{2} + \frac{2}{3},\frac{s}{2} + \frac{5}{6}}{\frac{1}{3},\frac{2}{3}}{ \lambda^2 y^{3}\; 2^4 3^3}\\[1.5em]
\hphantom{\mathcal{H}_{3,0}(\lambda;x,y)=}+ \sum_{s=0}^{\infty} \tfrac{\lambda^{s+1}}{(s+1)!} \; \left(
x^{3 (s+1)-2} y\left(x^3+6 x y\right)
+ x^{3 (s+1)-3}\binom{3 (s+1)-2}{1}6 y^2 \left(x^2+2 y\right)  
+ x^{3 (s+1)-3} \binom{3 (s+1)-2}{2}2! \cdot 12 y^3
\right.\\
\\
\hphantom{\mathcal{H}_{3,0}(\lambda;x,y)=}\quad \left. 
+ x^{3 (s+1)-5} \binom{3 (s+1)-2}{3} 3!\cdot 8 y^4
\right)\; \left(\tfrac{(3(s+1))!}{(3(s+1)- 2)! }\right)  \; \pFq{4}{2}{\frac{s+1}{2} + \frac{1}{6},\frac{s+1}{2} + \frac{1}{3},\frac{s+1}{2} + \frac{2}{3},\frac{s+1}{2} + \frac{5}{6}}{\frac{2}{3},\frac{4}{3}}{ \lambda^2 y^{3} \; 2^4 3^3}\\[1.5em]
\hphantom{\mathcal{H}_{3,0}(\lambda;x,y)=}+\sum_{s=0}^{\infty} \tfrac{\lambda^{s+2}}{(s+2)!}\; \left(
 x^{3 (s+2)-4}y^2 \left(x^3+6 x y\right)
 + x^{3 (s+2)-5}\binom{3 (s+2)-4}{1}6 y^3 \left(x^2+2 y\right)
 +x^{3 (s+2)-5} \binom{3 (s+2)-4}{2}2! \cdot 12 y^4 
\right.\\
\hphantom{\mathcal{H}_{3,0}(\lambda;x,y)=}\quad \left.
 + x^{3 (s+2)-7} \binom{3 (s+2)-4}{ 3}3!\cdot 8 y^5
\right) \; \left(\tfrac{(3(s+2))!}{(3(s+2) - 4)!2!}\right)  \;  \pFq{4}{2}{\frac{s+2}{2} + \frac{1}{6},\frac{s+2}{2} + \frac{1}{3},\frac{s+2}{2} + \frac{2}{3},\frac{s+2}{2} + \frac{5}{6}}{\frac{4}{3},\frac{5}{3}}{\lambda^2 y^{3} \; 2^4 3^3}\\[1.5em]
\pagebreak
\mathcal{H}_{4,1}(\lambda;x,y)=\sum_{s=0}^{\infty} \tfrac{\lambda^s}{s!} \; \left(
x^{4 s+1}
+ x^{4 s-1} \binom{4 s}{1}2 y
\right) \; \pFq{3}{1}{s+\frac{1}{4},s+\frac{1}{2},s+\frac{3}{4}}{\frac{1}{2}}{\lambda  y^2\; 2^6}\\[1.5em]
\hphantom{\mathcal{H}_{4,0}(\lambda;x,y)=}+\sum_{s=0}^{\infty} \tfrac{\lambda^{s+1}}{(s+1)!} \; \left(
x^{4 s+3}y
+ x^{4 s+1} \binom{4 s+2}{1}2 y^2
\right)\; \left(\tfrac{(4(s+1))!}{(4(s+1) -2)!}\right)\; \pFq{3}{1}{(s+1) + \frac{1}{4},(s+1) + \frac{1}{2},(s+1) + \frac{3}{4}}{\frac{3}{2}}{\lambda  y^2\; 2^6}\\[1.5em]
\mathcal{H}_{4,2}(\lambda;x,y)=\sum_{s=0}^{\infty} \tfrac{\lambda^s}{s!} \; \left(
x^{4 s} \left(x^2+2 y\right)
+ x^{4 s} \binom{4 s}{1}4 y
+ x^{4 s-2} \binom{4 s}{2}2!\cdot 4 y^2 
\right) \; \pFq{3}{1}{s+\frac{1}{4},s+\frac{1}{2},s+\frac{3}{4}}{\frac{1}{2}}{\lambda  y^2\; 2^6}\\[1.5em]
\hphantom{\mathcal{H}_{4,0}(\lambda;x,y)=}+\sum_{s=0}^{\infty} \tfrac{\lambda^{s+1}}{(s+1)!} \; \left(
x^{4 s+2}y \left(x^2+2 y\right)
+ x^{4 s+2} \binom{4 s+2}{1}4 y^2
+ x^{4 s} \binom{4 s+2}{2}2!\cdot 4 y^3
\right)\; \left(\tfrac{(4(s+1))!}{(4(s+1) -2)!}\right)\; \pFq{3}{1}{(s+1) + \frac{1}{4},(s+1) + \frac{1}{2},(s+1) + \frac{3}{4}}{\frac{3}{2}}{\lambda  y^2\; 2^6}\\[1.5em]
\mathcal{H}_{4,3}(\lambda;x,y)=\sum_{s=0}^{\infty} \tfrac{\lambda^s}{s!} \; \left(
x^{4 s} \left(x^3+6 x y\right)
+ x^{4 s-1}\binom{4 s}{1} 6 y\left(x^2+2 y\right)
+ x^{4 s-1} \binom{4 s}{2}2!\cdot 12 y^2 
+ x^{4 s-3} \binom{4 s}{3}3!\cdot 8 y^3
\right) \; \pFq{3}{1}{s+\frac{1}{4},s+\frac{1}{2},s+\frac{3}{4}}{\frac{1}{2}}{\lambda  y^2\; 2^6}\\[1.5em]
\hphantom{\mathcal{H}_{4,0}(\lambda;x,y)=}+\sum_{s=0}^{\infty} \tfrac{\lambda^{s+1}}{(s+1)!} \; \left(
x^{4 s+2} y\left(x^3+6 x y\right)
+ x^{4 s+1}\binom{4 s+2}{1}6 y^2 \left(x^2+2 y\right)
+ x^{4 s+1} \binom{4 s+2}{2}2!\cdot 12 y^3 
+ x^{4 s-1} \binom{4 s+2}{3}3!\cdot 8 y^4
\right)\;\cdot\\
\hphantom{\mathcal{H}_{4,0}(\lambda;x,y)=}\quad \cdot\; \left(\tfrac{(4(s+1))!}{(4(s+1) -2)!}\right)\; \pFq{3}{1}{(s+1) + \frac{1}{4},(s+1) + \frac{1}{2},(s+1) + \frac{3}{4}}{\frac{3}{2}}{\lambda  y^2\; 2^6}\\[1.5em]
\end{longtable}
\end{landscape}

\normalsize

\appendix
\numberwithin{equation}{section}

\section{Proofs}
\label{app:proofs}

\subsection{Proof of Lemma~\ref{lem:KD}}
\label{app:lemP}

\genKD*
\begin{proof}
Consider the presentation of an exponential generating function $\cG(\lambda;x,y)$ in the form as described in~\eqref{eq:EGFb},
\[
\cG(\lambda;x,y)=\sum_{r=0}^{\infty} x^r \sum_{m=0}^{\infty}\frac{\lambda^{r+m}}{(r+m)!}\; g_{r,m}(y)\,.
\]
Application of the lacunary dilatation operator $\bL_K$ (for $K\geq 1$) amounts according to the definition given in~\eqref{def:lacD} to the manipulation
\begin{equation}
\bL_{K}\left(\cG(\lambda;x,y)\right)=\sum_{r=0}^{\infty}x^r \sum_{m=0}^{\infty}
\frac{1}{(r+m)!}\; g_{r,m}(y)\; \delta_{(r+m)\Mod K,0}\frac{\lambda^{\frac{r+m}{K}}(r+m)!}{\left(\frac{r+m}{K}\right)!}\,.
\end{equation}
To proceed, we split the summation over $r$ into $\Mod K$ classes, where we need to take particular care to treat the special cases where $r=s\cdot K$ separately (see ensuing arguments below):
\begin{subequations}
\begin{align}
\bL_K(\cG(\lambda;x,y))&=\sum_{s=0}^{\infty}x^{s\cdot K}\sum_{m=0}^{\infty}\frac{\lambda^{s+\frac{m}{K}}}{(s+\frac{m}{K})!}\;g_{s\cdot K,m}(y)\;\delta_{m\Mod K,0}\label{eq:lemDproofA}\\
&\quad +\sum_{\alpha=1}^{K-1}\sum_{s=0}^{\infty} x^{s\cdot K+\alpha}\sum_{m=0}^{\infty}\frac{\lambda^{s+\frac{\alpha+m}{K}}}{(s+\frac{\alpha+m}{K})!}\; g_{s\cdot K+\alpha,m}(y)\; \delta_{(\alpha+m)\Mod K,0}\label{eq:lemDproofB}
\end{align}
\end{subequations}
Here, we made use of the standard convention $\sum_{\alpha=1}^{K-1}\dotsc=0$ for $K=1$, and we emphasized (as in the statement of the Lemma) the positions where splitting into $\Mod K$ classes has been performed via using the notation $s\cdot K$. The remaining constraint in~\eqref{eq:lemDproofA},
\[
m\Mod K=0\,,
\]
is solved evidently by $m=q\cdot K$ (for $q\in \bZ_{\geq 0}$), while the remaining constraint in~\eqref{eq:lemDproofB},
\[
(\alpha+m)\Mod K=0\,,
\]
is solved by\footnote{Equivalently, we could solve the constraint in the form $m=q\cdot K-\alpha$, but then for $q\in \bZ_{\geq 1}$, since the original summation over $m$ is taken for $m\in \bZ_{\geq 0}$.} $m=q\cdot K+K-\alpha$ (for a given $\alpha$, and with $q\in \bZ_{\geq 0}$).\\

Finally, purely for reasons of convenience in later applications, we chose to re-index the summation over $\alpha$ by making the replacement $\alpha\mapsto K-\alpha$, which gives the form of the formula as presented in the statement of the Lemma.
\end{proof}

\subsection{Proof of Corollary~\ref{cor:KD}}
\label{app:corDproof}

\genKDcor*
\begin{proof}
While in principle a rather elementary computation, the proof of the statement of the corollary requires a somewhat intricate resolution of the constraints implied by subdividing the summations involved in the calculation of $\cG_{K,L}(\lambda;x,y)$ as described in Lemma~\ref{lem:KD} further (into sums over even and odd second indices of the coefficients $g_{r,m}$, respectively). The cases of even and odd $K$ must be treated separately, as will become evident shortly.\\

Consider then first the case $K=2T$ (for $T\in \bZ_{\geq 1}$). Specializing the general formula for $\cG_{K,0}(\lambda;x,y)$ as given in~\eqref{eq:KLGF} to this case, 
\begin{subequations}
\begin{align}
\bL_{K=2T}\left(\cG(\lambda;x,y)\right)&=\sum_{s=0}^{\infty}x^{s\cdot K}\sum_{q=0}^{\infty}\frac{\lambda^{s+q}}{(s+q)!}\; g_{s\cdot K,q\cdot K}(y)
\label{eq:corDproofEA}\\
&\quad +\sum_{\alpha=1}^{K-1}\sum_{s=0}^{\infty} x^{(s+1)\cdot K-\alpha}
\sum_{q=0}^{\infty}\frac{\lambda^{s+q+1}}{(s+q+1)!}\; g_{(s+1)\cdot K,q\cdot K+\alpha}\,,\label{eq:corDproofEB}
\end{align}
\end{subequations}
we recognize immediately that since $K=2T$ is even, so is $q\cdot K$ for all $q\in \bZ_{\geq 0}$. Consequently, all terms of~\eqref{eq:corDproofEA} contribute to the ``even part'' (marked $\text{part $E$}$ in~\eqref{eq:KLGFeven}), as well as all those terms of~\eqref{eq:corDproofEB} for which $\alpha=2\beta$ (and whence we obtain a reduction of the respective sum over $\alpha$, with $1\leq \alpha \leq K-1=2T-1$, to a summation over $\beta$ with $1\leq \beta\leq T-1$). For the ``odd part'' (marked $\text{part $O$}$ in~\eqref{eq:KLGFeven}), none of the terms in~\eqref{eq:corDproofEA} contribute (since evidently all second indices $m$ of $g_{r,m}(y)$ in this summation are even), with non-trivial contributions only from terms of~\eqref{eq:corDproofEB} with $\alpha=2\beta-1$ (resulting in a restriction of the respective summation over $\alpha$, with $1\leq \alpha\leq K-1=2T-1$, to a summation over $\beta$ with $1\leq \beta \leq T$). This proves the statement of the corollary for $K=2T$.\\

The proof for the case $K=2T+1$ (with $T\in \bZ_{\geq 1}$) is entirely analogous, yet the separation of the respective summations over contributions with even or odd second indices of the coefficients $g_{r\!,\,m}(y)$ is more involved. Suffice it here to present the resolution of the respective constraints in terms of the admissible forms for summation indices $r$ and $q$ as well as the summation ranges for the respective auxiliary summation indices $\beta$, from which the second part of the corollary as given by~\eqref{eq:KLGFodd} follows (with $0\leq s,\ell\leq \infty$):
\begin{equation}
\begin{array}{llll}
\text{part} & r & q & \text{range for $\beta$}\\
\hline
\text{part $E$} & r=s\cdot K & q=2\ell & -\\
& r=(s+1)\cdot K-2\beta & q=2\ell & 1\leq \beta\leq T\\
& r=(s+1)\cdot K-2\beta+1 & q=2\ell+1 & 1\leq \beta \leq T\\
\hline
\text{part $O$} & r=s\cdot K & q=2\ell+1 & -\\
& r=(s+1)\cdot K-2\beta+1 & q=2\ell & 1\leq \beta\leq T\\
& r=(s+1)\cdot K-2\beta & q=2\ell+1 & 1\leq \beta\leq T\\
\end{array}
\end{equation}
\end{proof}

\subsection{Details on the semi-linear normal ordering technique used in the proof of Proposition~\ref{prop:LshiftHLGF}}
\label{app:SLNO}

For the interested readers' convenience, we briefly recall the salient details of the so-called semi-linear normal ordering technique, the origins of which date back at least to the 1960s (see~\cite{dattoli1997evolution} for a historical overview and an extensive list of application examples in physics, and also~\cite{blasiak2005boson,blasiak2011combinatorial,bdp2017} for applications of the technique in combinatorics and chemistry).\\

Let thus $D$ be a differential operator of the ``semi-linear form'' (i.e.\ at most linear in $\tfrac{d}{dx}$),
\begin{equation}
D=q(x)\tfrac{d}{dx}+v(x)\,,
\end{equation}
where for simplicity (and as sufficient for the purposes of the proof of Proposition~\ref{prop:LshiftHLGF}) we assume that $q(x)$ and $v(x)$ are polynomials with real-valued coefficients, $q(x),v(x)\in \bR[x]$. Let $f(x)$ be an entire function, and $\lambda$ a formal variable. Then one finds that
\begin{equation}
e^{\lambda D}f(x)=g(\lambda;x) f(T(\lambda;x))\,,
\end{equation}
where the \emph{substitution function} $T(\lambda;x)$ and the \emph{prefactor function} $g(\lambda;x)$ are computed via solving the following initial value problem:
\begin{equation}
\begin{aligned}
\tfrac{\partial}{\partial \lambda}T(\lambda;x)&=q(T(\lambda;x))\,,\quad & T(0;x)&=x\\
\tfrac{\partial}{\partial \lambda} \ln (g(\lambda;x))&=v(T(\lambda;x))\,,\quad &
g(0;x)&=1\,.
\end{aligned}
\end{equation}

\subsection{Details of the proof of Corollary~\ref{cor:HPlacKL}}
\label{app:proofLacKL}

\HPlacKL*
\begin{proof}
In order to efficiently apply the operational rule stated in~\eqref{eq:propHLGF} to the explicit formulae for $\cH_{K,0}(\lambda;x,y)$ as stated in Theorem~\ref{thm:HPlacK}, recall first that since $\cH_{1,0}(\mu;x,y)$ is the exponential generating function of the two-variable Hermite polynomials $H_n(x,y)$, we find by definition that
\begin{equation}\label{eq:appKLlacA}
\left[\left(\frac{\partial}{\partial \mu}\right)^r \cH_{1,0}(\mu;x,y)\right]\bigg\vert_{\mu\mapsto 0}=H_r(x,y)\,.
\end{equation}
We also need the following elementary identity (where $f(\mu;x,y)$ and $g(\mu;x,y)$ are formal power series in the formal variables $\lambda,x,y$)
\begin{equation}\label{eq:appKLlacB}
\left(\frac{\partial}{\partial \mu}\right)^L\left(f(\mu;x,y)g(\mu;x,y)\right)=\sum_{q=0}^L\binom{L}{q}
\left(\left(\frac{\partial}{\partial \mu}\right)^{L-q}f(\mu;x,y)\right)\left(\left(\frac{\partial}{\partial \mu}\right)^{q}g(\mu;x,y)\right)\,.
\end{equation}
To complete the proof, note that according to~\eqref{eq:propHLGF} in the calculation of $\cH_{K,L}(\lambda;x,y)$ from $\cH_{K,0}(\lambda;x,y)$, the variable $x$ is replaced by $(x+2\mu y)$, and the entire expression for $\cH_{K,)}(\lambda;x+2\mu y,y)$ multiplied by $\cH_{1,0}(\mu;x,y)$, followed by taking $L$ derivatives of the overall expression with respect to $\mu$ and setting $\mu$ to zero, in the process obtaining $\cH_{K,L}(\lambda;x,y)$. It thus suffices to focus on a generic term of the form $x^P$ (for $P\in \bZ_{\geq0}$) that occurs in $\cH_{K,0}(\lambda;x,y)$. It transformed via the previously described procedure as follows:
\begin{equation}
\begin{aligned}
&\left[\left(\frac{\partial}{\partial \mu}\right)^L\; \cH_{1,0}(\mu;x,y)\; x^P\right]\bigg\vert_{\mu\mapsto 0}\\
&\qquad\overset{\eqref{eq:appKLlacB}}{=}\left[\sum_{q=0}^L\binom{L}{q}
\left(\left(\frac{\partial}{\partial \mu}\right)^{L-q}\cH_{1,0}(\mu;x,y)\right)
\left(\left(\frac{\partial}{\partial \mu}\right)^{q}x^P\right)\right]\bigg\vert_{\mu\mapsto 0}\\
&\qquad\overset{\eqref{eq:appKLlacA}}{=}\sum_{q=0}^L \binom{L}{q}H_{L-q}(x,y)\;q! \binom{P}{q} \;x^{P-q}\; (2y)^q\,.
\end{aligned}
\end{equation}
The proof then follows by application of this auxiliary formula to the computation of $\cH_{K,L}(\lambda;x,y)$ as described above.
\end{proof}

\section{Maple code listing for algorithmic verification of lacunary generating functions}
\label{app:Maple}

As a consistency check, we provide a some listings of \textsc{Maple}${}^{TM}$ code here for the readers' convenience.\\

Bi-variate Hermite Polynomials of parameter $M=2,3,..$ are denoted $H(n,M,x,y)$:
\begin{lstlisting}
> H:=proc(n,M,x,y)n!*sum(x^(n-M*r)*y^r/(r!*(n-M*r)!),r=0..floor(n/M)) ;end;
\end{lstlisting}
We are interested only in the case $M=2$; presented below are the traditional two-variable Hermite polynomials, denoted $He(n, x,y)$:
\begin{lstlisting}
> He:=proc(n,x,y)H(n,2,x,y);end;
\end{lstlisting}
The polynomials $He(n,x,y)$ are related to the conventional form of the Hermite polynomials via (Maple notation)
\begin{equation}
\text{\lstinline{HermiteH(n,x)}}= \text{\lstinline{He(n, 2*x ,-1)}}\,.
\end{equation}
This may be explicitly verified via Maple, taking differences:
\begin{lstlisting}
> seq(expand(HermiteH(kk,x))-He(kk,2*x,-1),kk=0..7);
\end{lstlisting}

We present an algorithm to verify the lacunary generating function formulae by means of taking derivatives with respect to the formal parameter $\lambda$ of the generating functions $\cH_{K,L}(\lambda;x,y)$ followed by setting $\lambda$ to zero, thus producing the respective Hermite polynomials. In the following code listings, the formal variable is denoted \lstinline{Lam}, while the triple lacunary generating function $\cH_{3,0}(\lambda;x,y)$ is denoted \lstinline{GHP3_0(smax,Lam,x,y)}. Here, the variable \lstinline{smax} is the order of truncation of powers of \lstinline{Lam}.
\begin{lstlisting}
>GHP3_0:=proc(smax,Lam,x,y)sum(Lam^s*x^(3*s)*hypergeom([s/2+1/6, s/2+1/3,s/2+2/3,s/2+5/6],[1/3,2/3],Lam^2*y^3*(2^4)*(3^3))/s!,s=0.. smax)+sum(Lam^(s+1)*x^(3*s+1)*y*((3*s+3)!/(3*s+1)!)*hypergeom([ (s+1)/2+1/6,(s+1)/2+1/3,(s+1)/2+2/3,(s+1)/2+5/6],[2/3,4/3],Lam^2* y^3*(2^4)*(3^3))/(s+1)!,s=0..smax)+sum(Lam^(s+2)*x^(3*s+2)*y^2*((3* s+6)!/(3*s+2)!)*hypergeom([(s+2)/2+1/6,(s+2)/2+1/3,(s+2)/2+2/3, (s+2)/2+5/6],[4/3,5/3],Lam^2*y^3*(2^4)*(3^3))/(2*(s+2)!),s=0..smax) ;end;
\end{lstlisting}
For a concrete test of the triple lacunary generating function, we truncate at \lstinline{smax=100}. We compare the difference between the \lstinline{kk}-th derivative with respect to to \lstinline{Lam} evaluated at \lstinline{Lam=0} and the two-variable Hermite polynomial of order \lstinline{3*kk}, for \lstinline{kk=1..16}:
\begin{lstlisting}
> seq(expand(simplify(subs(Lam=0,diff(GHP3_0(100,Lam,x,y),Lam$kk))))- He(3*kk,x,y),kk=1..16);
\end{lstlisting}

Below we repeat the same procedure for the quadruple lacunary generating function $\cH_{4,0}(\lambda;x,y)$, denoted \lstinline{GHP4_0(smax, Lam, x,y)}:
\begin{lstlisting}
>GHP4_0:=proc(smax,Lam,x,y)sum(Lam^s*x^(4*s)*hypergeom([s+1/4,s+1/2, s+3/4],[1/2],Lam*y^2*2^6)/s!,s=0..smax)+sum(Lam^(s+1)*x^(4*s+2)*y*( (4*s+4)!/(4*s+2)!)*hypergeom([s+1+1/4,s+1+1/2,s+1+3/4],[3/2],Lam* y^2*2^6)/(s+1)!,s=0..smax);end;
\end{lstlisting}
Truncating at \lstinline{smax=100}, we compare the differences between the derivatives with respect to \lstinline{Lam} evaluated at \lstinline{Lam=0} and the two-variable Hermite polynomials \lstinline{He(4*kk,x,y)} of order \lstinline{4*kk}, for \lstinline{kk=1..16}:
\begin{lstlisting}
> seq(expand(simplify(subs(Lam=0,diff(GHP4_0(100,Lam,x,y),Lam$kk))))- He(4*kk,x,y),kk=1..16);
\end{lstlisting}

As a final example, we present below an analogous computation for the quintuple lacunary generating function $\cH_{5,0}(\lambda;x,y)$, with an example calculation for truncation order \lstinline{smax=150} and \lstinline{kk=1..16} (whence verifying up to the polynomials of order $5*15$):

\begin{lstlisting}
> GHP5_0:=proc(smax,Lam, x,y) sum(Lam^s*x^(5*s)*hypergeom([s/2+1/10, s/2+1/5,s/2+3/10,s/2+2/5,s/2+3/5,s/2+7/10,s/2+4/5,s/2+9/10],[1/5, 2/5,3/5,4/5],Lam^2*y^5*2^8*5^5)/s!,s=0..smax)+
sum(Lam^(s+1)*x^(5*s+3)*y*((5*s+5)!/(5*s+3)!)*hypergeom(subs (s=s+1,[s/2+1/10,s/2+1/5,s/2+3/10,s/2+2/5,s/2+3/5,s/2+7/10,s/2+4/5, s/2+9/10]),[2/5,3/5,4/5,6/5],Lam^2*y^5*2^8*5^5)/(s+1)!,s=0..smax)+
sum(Lam^(s+1)*x^(5*s+1)*y^2*((5*s+5)!/(2*(5*s+1)!))* hypergeom(subs(s=s+1,[s/2+1/10,s/2+1/5,s/2+3/10,s/2+2/5,s/2+3/5, s/2+7/10,s/2+4/5,s/2+9/10]),[3/5,4/5,6/5,7/5],Lam^2*y^5*2^8*5^5)/ (s+1)!,s=0..smax)+ sum(Lam^(s+2)*x^(5*s+4)*y^3*((5*s+10)!/(6* (5*s+4)!))*hypergeom(subs(s=s+2,[s/2+1/10,s/2+1/5,s/2+3/10,s/2+2/5, s/2+3/5,s/2+7/10,s/2+4/5,s/2+9/10]),[4/5,6/5,7/5,8/5],Lam^2*y^5* 2^8*5^5)/(s+2)!,s=0..smax) + sum(Lam^(s+2)*x^(5*s+2)*y^4*((5* s+10)!/(24*(5*s+2)!))*hypergeom(subs(s=s+2,[s/2+1/10,s/2+1/5, s/2+3/10,s/2+2/5,s/2+3/5,s/2+7/10,s/2+4/5,s/2+9/10]),[6/5,7/5,8/5, 9/5],Lam^2*y^5*2^8*5^5)/(s+2)!,s=0..smax);end;

> seq(expand(simplify(subs(Lam=0,diff(GHP5_0(150,Lam,x,y),Lam$kk))))- He(5*kk,x,y),kk=1..16);
\end{lstlisting}


\begin{thebibliography}{10}

\bibitem{babusci2017lacunary}
D.~Babusci, G.~Dattoli, K.~G{\'o}rska, and K.A.~Penson.
\newblock Lacunary generating functions for the Laguerre polynomials.
\newblock {\em S{\'e}minaire Lotharingien de Combinatoire}, 76:B76b, 2017.

\bibitem{penson2018quasi}
K.A.~Penson, K.~G{\'o}rska, A.~Horzela, and G.~Dattoli.
\newblock Quasi-relativistic heat equation via L{\'e}vy stable distributions: Exact solutions.
\newblock {\em Annalen der Physik}, 530(3):1700374, 2018.

\bibitem{babusci2010lectures}
D.~Babusci, G.~Dattoli, and M.~Del~Franco.
\newblock Lectures on mathematical methods for physics.
\newblock {\em Technical Report}, 58, 2010
\newblock \url{http://opac22.bologna.enea.it/RT/2010/2010_58_ENEA.pdf}.

\bibitem{beals2016special}
R.~Beals and R.~Wong.
\newblock {\em Special functions and orthogonal polynomials}, Cambridge Studies in Advanced Mathematics, vol.~153.
\newblock Cambridge University Press, 2016.

\bibitem{kampe}
P.~Appell and J.~Kamp\'{e}~de F\'{e}riet.
\newblock {\em {Fonctions hyperg\'{e}om\'{e}triques et hypersph\'{e}riques: polyn\^{o}mes d'Hermite}}.
\newblock Paris : Gauthier-Villars, 1926.

\bibitem{dattoli1997evolution}
G.~Dattoli, P.L.~Ottaviani, A.~Torre, and L.~V{\'a}zquez.
\newblock Evolution operator equations: Integration with algebraic and finite difference methods: Applications to physical problems in classical and quantum mechanics and quantum field theory.
\newblock {\em La Rivista del Nuovo Cimento (1978-1999)}, 20(2):3, 1997.


\bibitem{prudnikov1992integrals}
A.P.~Prudnikov, Y.~A.~Brychkov, and
  O.I.~Marichev.
\newblock {\it Integrals and series},
\newblock vol.~2,
\newblock Gordon and Breach Science Publishers, New York, 1992.

\bibitem{foataStrehl1984}
D.~Foata and V.~Strehl.
\newblock {Combinatorics of Laguerre polynomials}.
\newblock In {\em {Enumeration and Design, Waterloo Jubilee Conference}}, pages
  123--140. Academic Press, 1984.

\bibitem{dattoli2017operational}
G.~Dattoli, E.~Di~Palma, E.~Sabia, K.~Gorska, A.~Horzela, and K.A.~Penson.
\newblock {Operational versus umbral methods and the Borel transform}.
\newblock {\em International Journal of Applied and Computational Mathematics},
  3(4):3489--3510, 2017.

\bibitem{strehl2017lacunary}
V.~Strehl.
\newblock Lacunary Laguerre series from a combinatorial perspective.
\newblock {\em S{\'e}minaire Lotharingien de Combinatoire}, 76:B76c, 2017.

\bibitem{foata1981some}
D.~Foata.
\newblock {Some Hermite polynomial identities and their combinatorics}.
\newblock {\em Advances in Applied Mathematics}, 2(3):250--259, 1981.

\bibitem{gessel2005triple}
I.M.~Gessel and P.~Jayawant.
\newblock A triple lacunary generating function for Hermite polynomials.
\newblock {\em the electronic journal of combinatorics}, 12(1):30, 2005.

\bibitem{gorskalacunary}
K.~G{\'o}rska, K.A.~Penson, and G.~Dattoli.
\newblock Lacunary generating functions for Legendre and Chebyshev polynomials
\newblock {\it (unpublished)}.

\bibitem{nieto1995arbitrary}
M.M.~Nieto and D.R.~Truax.
\newblock Arbitrary-order Hermite generating functions for obtaining arbitrary-order coherent and squeezed states.
\newblock {\em Physics Letters A}, 208(1-2):8--16, 1995.

\bibitem{dattoli1998operational}
G.~Dattoli, A.~Torre, and M.~Carpanese.
\newblock Operational rules and arbitrary order Hermite generating functions.
\newblock {\em Journal of mathematical analysis and applications},
  227(1):98--111, 1998.


\bibitem{reutenauer}
C.~Reutenauer.
\newblock Free Lie algebras, vol.~7, {\it London Mathematical Society
  Monographs, New Series}, 1993.

\bibitem{berstel}
J.~Berstel and C.~Reutenauer.
\newblock {\em Rational series and their languages}, vol.~12.
\newblock Springer-Verlag, 1988.

\bibitem{GT1}
N.~Bourbaki.
\newblock {\em {General Topology: Chapters III}}, vol.~18,
\newblock Springer Science \& Business Media, 2013.

\bibitem{SMF}
G.H.E.~Duchamp, K.A.~Penson, and C.~Tollu.
\newblock {Physique combinatoire I: Groupes \`{a} un param\`{e}tre}.
\newblock {\em Gazette des math\'{e}maticiens}, (130):37--50, 2011.

\bibitem{treves}
F.~Treves.
\newblock {\em {Topological Vector Spaces, Distributions and Kernels: Pure and Applied Mathematics}}, vol.~25,
\newblock Elsevier, 2016.

\bibitem{crofton79}
M.W.~Crofton.
\newblock {Theorems in the calculus of operations}.
\newblock {\em Q. J. Math.}, 16:323--352, 1879.

\bibitem{dattoli2000generalized}
G.~Dattoli.
\newblock Generalized polynomials, operational identities and their applications.
\newblock {\em Journal of Computational and Applied mathematics},
  118(1-2):111--123, 2000.

\bibitem{NIST:DLMF}
{\it NIST Digital Library of Mathematical Functions}.
\newblock http://dlmf.nist.gov/, Release 1.0.18 of 2018-03-27.
\newblock F.W.J.~Olver, A.B.~{Olde Daalhuis}, D.W.~Lozier, B.I.~Schneider,
  R.F.~Boisvert, C.W.~Clark, B.R.~Miller and B.V.~Saunders, eds.

\bibitem{blasiak2005boson}
P.~Blasiak, A.~Horzela, K.A.~Penson, G.H.E.~Duchamp, and A.I.~Solomon.
\newblock {Boson normal ordering via substitutions and Sheffer-type polynomials}.
\newblock {\em Physics Letters A}, 338(2):108--116, 2005.

\bibitem{blasiak2011combinatorial}
P.~Blasiak and P.~Flajolet.
\newblock Combinatorial models of creation-annihilation.
\newblock {\em S{\'e}minaire Lotharingien de Combinatoire}, 65(B65c):1--78,
  2011.

\bibitem{bdp2017}
N.~Behr, G.H.E.~Duchamp, and K.A.~Penson.
\newblock {Combinatorics of chemical reaction systems}.
\newblock {\em arXiv:1712.06575}, 2017.

\end{thebibliography}
\end{document}